\documentclass[acmsmall,screen,nonacm]{acmart}

\usepackage{amsmath}
\usepackage{algorithm}
\usepackage{algpseudocode}
\usepackage{booktabs}
\usepackage{listings}
\usepackage{tabularx}
\usepackage{xcolor}
\usepackage{colortbl}
\usepackage{xspace}
\usepackage{microtype}
\usepackage{needspace}
 
\AddToHook{env/definition/before}{\Needspace{6\baselineskip}}
\AddToHook{env/proposition/before}{\Needspace{6\baselineskip}}

\makeatletter
\newenvironment{formulationstyle}{%
  \par
  \setlength{\parindent}{\normalparindent}%
  \setlength{\parskip}{3pt}%
  \let\@afterindentfalse\@afterindenttrue
  \renewcommand{\subsubsection}{\def\@toclevel{3}%
    \@startsection{subsubsection}{3}{\parindent}%
      {\z@}{-3.5\p@}%
      {\ACM@NRadjust{\normalfont\bfseries\@adddotafter}}}%
  \renewcommand{\@acmdefinitionheadfont}{\normalfont\bfseries}%
  \renewcommand{\@acmdefinitionnotefont}{\normalfont\bfseries}%
  \renewcommand{\@acmplainbodyfont}{\normalfont}%
  \renewcommand{\@acmplainheadfont}{\normalfont\bfseries}%
  \renewcommand{\@acmplainnotefont}{\normalfont\bfseries}%
  \let\rondo@definitionstyle\th@acmdefinition
  \def\th@acmdefinition{\rondo@definitionstyle
    \thm@preskip\parskip \thm@postskip\z@skip}%
  \let\rondo@plainstyle\th@acmplain
  \def\th@acmplain{\rondo@plainstyle
    \thm@preskip\parskip \thm@postskip\z@skip}%
  \renewenvironment{proof}[1][\proofname]{%
    \par\pushQED{\qed}%
    \normalfont\topsep\z@\partopsep\z@
    \trivlist
    \item[\indent\hskip\labelsep
      {\normalfont\bfseries ##1\@addpunct{.}}]\ignorespaces
  }{\popQED\endtrivlist\@endpefalse}%
}{\par}

\makeatother

\newcolumntype{Y}{>{\raggedright\arraybackslash}X}
\newcolumntype{Z}{>{\raggedleft\arraybackslash}X}
\newcommand{\rondo}{\textsc{Rondo}\xspace}

\definecolor{TableHeader}{HTML}{ECEFF1}
\definecolor{TableStripe}{HTML}{F5F6F7}
\definecolor{TableHighlight}{HTML}{EAF1F7}
\definecolor{TableBaseline}{HTML}{F7F1E7}
\newcommand{\rondotablestyle}{%
  \centering
  \small
  \setlength{\tabcolsep}{4pt}%
  \renewcommand{\arraystretch}{1.16}%
}

\newcommand{\rondobox}[2]{%
  \par\smallskip
  \begingroup
  \setlength{\fboxsep}{5pt}%
  \setlength{\fboxrule}{0.5pt}%
  \noindent\fcolorbox{black!20}{black!3}{%
    \parbox{\dimexpr\linewidth-2\fboxsep-2\fboxrule\relax}{%
      \textbf{#1} #2}}%
  \par\endgroup\smallskip
}
\newcommand{\rqanswer}[2]{\rondobox{Answer to RQ#1.}{#2}}
\newcommand{\observation}[2]{\rondobox{Observation #1.}{#2}}

\definecolor{MotivatingBlue}{HTML}{245A83}
\definecolor{MotivatingRed}{HTML}{A32121}
\definecolor{MotivatingGray}{HTML}{5F6368}
\definecolor{MotivatingRule}{HTML}{AEB4BA}
\definecolor{MotivatingBg}{HTML}{F6F8FA}
\definecolor{MotivatingErrorBg}{HTML}{FCE8E8}
\definecolor{MethodGreen}{HTML}{176B45}
\definecolor{MethodGreenBg}{HTML}{EFF8F2}
\definecolor{MethodGreenRule}{HTML}{90BFA1}
\definecolor{MethodRedBg}{HTML}{FFF1F1}
\definecolor{MethodRedRule}{HTML}{D7A0A0}
\definecolor{AlgorithmBlue}{HTML}{0000CC}
\algrenewcommand{\algorithmiccomment}[1]{\hfill\textcolor{AlgorithmBlue}{$\triangleright$~#1}}
\lstdefinestyle{motivation}{
  language=Python,
  basicstyle=\ttfamily\scriptsize,
  keywordstyle=\color{MotivatingBlue}\bfseries,
  commentstyle=\color{MotivatingGray},
  stringstyle=\color{MotivatingBlue},
  numbers=left,
  numberstyle=\tiny\color{MotivatingGray},
  numbersep=5pt,
  frame=single,
  framerule=0.35pt,
  rulecolor=\color{MotivatingRule},
  backgroundcolor=\color{MotivatingBg},
  xleftmargin=1.4em,
  framexleftmargin=1.1em,
  aboveskip=3pt,
  belowskip=0pt,
  breaklines=true,
  columns=fullflexible,
  keepspaces=true,
  showstringspaces=false
}

\setcopyright{none}

\title{Stateful Testing for Deep Learning Compilers}

\author{Jiatao Ge}
\orcid{0009-0009-9394-8100}
\affiliation{%
  \institution{State Key Laboratory for Novel Software Technology, Nanjing University}
  \city{Nanjing}
  \country{China}}
\email{gejiatao@smail.nju.edu.cn}

\author{Ruixiang Qian}
\orcid{0009-0003-5040-3123}
\affiliation{%
  \institution{State Key Laboratory for Novel Software Technology, Nanjing University}
  \city{Nanjing}
  \country{China}}
\email{qianrx@smail.nju.edu.cn}

\author{Yibo Hu}
\orcid{0009-0009-4103-1421}
\affiliation{%
  \institution{State Key Laboratory for Novel Software Technology, Nanjing University}
  \city{Nanjing}
  \country{China}}
\email{yibohu@smail.nju.edu.cn}

\author{Yanzhou Mu}
\orcid{0000-0003-1816-2246}
\affiliation{%
  \institution{UNIST}
  \city{Ulsan}
  \country{Republic of Korea}}
\email{602022320006@smail.nju.edu.cn}

\author{Yinglong Zou}
\orcid{0009-0006-9375-7417}
\affiliation{%
  \institution{State Key Laboratory for Novel Software Technology, Nanjing University}
  \city{Nanjing}
  \country{China}}
\email{652023320004@smail.nju.edu.cn}

\author{Jiawei Liu}
\orcid{0000-0002-4930-9637}
\affiliation{%
  \institution{State Key Laboratory for Novel Software Technology, Nanjing University}
  \city{Nanjing}
  \country{China}}
\email{jwliu@nju.edu.cn}

\author{Zhenyu Chen}
\orcid{0000-0002-9592-7022}
\affiliation{%
  \institution{State Key Laboratory for Novel Software Technology, Nanjing University}
  \city{Nanjing}
  \country{China}}
\email{zychen@nju.edu.cn}

\author{Chunrong Fang}
\orcid{0000-0002-9930-7111}
\authornote{Corresponding author.}
\affiliation{%
  \institution{State Key Laboratory for Novel Software Technology, Nanjing University}
  \city{Nanjing}
  \country{China}}
\email{fangchunrong@nju.edu.cn}

\renewcommand{\shortauthors}{Ge et al.}

\makeatletter
\patchcmd{\@mktitle@i}{\par\bigskip}{\par\vskip8pt}{}%
  {\PackageError{rondo}{Could not adjust the title spacing}{}}
\patchcmd{\@mkauthors@i}{\unvbox\mktitle@bx\par\medskip}%
  {\unvbox\mktitle@bx\par\vskip2pt}{}%
  {\PackageError{rondo}{Could not adjust the title-author spacing}{}}
\patchcmd{\@mkauthors@i}{\par\medskip}{\par\vskip6pt}{}%
  {\PackageError{rondo}{Could not adjust the author-abstract spacing}{}}
\pretocmd{\@specialsection}{\setlength{\medskipamount}{4pt plus .5pt minus .5pt}}{}%
  {\PackageError{rondo}{Could not adjust the metadata spacing}{}}
\makeatother
\AtBeginDocument{%
  \raggedbottom
  \AddToHookNext{shipout/after}{\flushbottom}}


\hypersetup{%
  keeppdfinfo,
  pdfauthor={Jiatao Ge, Ruixiang Qian, Yibo Hu, Yanzhou Mu, Yinglong Zou, Jiawei Liu, Zhenyu Chen, Chunrong Fang},
  pdfkeywords={deep learning compiler testing, mutation testing}}

\begin{document}

\begin{abstract}
Modern deep learning compilers compile models lazily and cache compiled
artifacts for reuse across invocations. Compiler state is information
retained by compilers across invocations that can affect
later invocations. Examples include cached code and the guards for reusing it.
Compiler state from earlier
invocations can cause an incorrect result in a later invocation. To
understand these stateful compiler bugs, we analyze 26 publicly reported
issues in PyTorch, JAX, and TensorFlow that describe such bugs. The analysis
shows that detecting these bugs requires two steps, establishing the
relevant compiler state through earlier invocations, and making a later
invocation that exposes the bug. To meet this requirement, we represent
each test case as
$\langle\textit{model},\textit{invocation sequence}\rangle$.
A fixed initial state and invocation sequence determine the resulting
compiler state. Based on this representation, we develop \rondo, a stateful
testing method for deep learning compilers. \rondo generates test cases by
mutating invocation sequences. For each test case, it runs the
invocation sequence on the model and records compilation behavior. A
Markov chain Monte Carlo (MCMC) sampler uses this behavior as feedback to
guide generation, favoring rarely observed
compilation behaviors. \rondo identifies possible bugs through three-way
differential testing, comparing results without compilation, compiled
results without compiler state from earlier invocations, and compiled
results with that state. A stateful bug is confirmed when the same
invocation produces an incorrect result with compiler state retained and a
correct result after that state is reset. Lastly, \rondo attempts to shorten
the invocation sequence while preserving the original bug to obtain a
smaller script that reproduces it. We evaluate \rondo on four representative
deep learning frameworks using 24 common model architectures. We detect
117 bugs, including 35 stateful bugs, and submit 39 bug reports, including
18 for stateful bugs. Developers confirm 35 reports, including 16 for
stateful bugs, and fix 18 bugs, including 11 stateful bugs. \rondo achieves
the highest behavior coverage and compiler Python source-line coverage
in the baseline comparison. Its bug checks achieve 99.26\% alarm precision.
It shortens 44 of 117 invocation sequences while preserving the original
bugs, reducing the lengths of these 44 sequences by 83.34\% on average.
\end{abstract}


\maketitle

\section{Introduction}
\label{sec:introduction}

Deep learning compilers are part of deep learning frameworks and provide
important infrastructure for running models~\cite{Arora2025TensorRight,Liu2024FunctionalTensorCompiler,Ikarashi2022Exo,Ikarashi2025Exo2}.
These compilers compile models lazily and cache compiled artifacts for reuse
across invocations. For example, PyTorch's compiler caches compiled artifacts
for later reuse~\cite{Ansel2024PyTorch2}. Similarly, MindSpore's compiler
caches compiled graphs for later model invocations. These caching mechanisms
improve efficiency and make the compilers stateful. \emph{Compiler state} is
information retained by compilers across invocations that can
affect later invocations.
Compiler state established by prior invocations can cause an incorrect
result in a later invocation, even when the same invocation is correct in
isolation. We call these bugs
\emph{stateful compiler bugs}.

Existing deep learning compiler testing methods aim to detect framework bugs
by generating diverse models and model inputs~\cite{Pham2019CRADLE,Liu2023NNSmith}.
Most of these methods test individual model invocations. We write such a
test case as a pair $\langle\textit{model},\textit{invocation}\rangle$, where
the model describes the computation under test and the invocation
specifies the model inputs. This
representation captures the current invocation but does not explicitly capture
the compiler state established by prior invocations.
To understand stateful compiler bugs, we analyze 26
publicly reported bug issues in PyTorch, JAX, and TensorFlow. The analysis
shows that these bugs can produce incorrect outputs and affect training.
Detecting these bugs requires generating invocation sequences
that establish compiler state and expose bugs in later invocations.
This requirement gives rise to three challenges.

\textbf{Challenge 1. Bug triggers across invocations.}
For a stateful bug, the invocation that creates the relevant compiler
state and the invocation that exposes the bug are different. Representing
a test case as $\langle\textit{model},\textit{invocation}\rangle$ does not explicitly
capture these invocations and their order. For example, invocation $A$ may
establish compiler state that causes a later invocation $B$ to produce an
incorrect result. Testing $A$ and $B$ separately would miss this bug.
Detecting it therefore requires explicitly searching invocation sequences
such as $\langle A,B\rangle$~\cite{Deng2025Coni,Liang2023SequenceFuzzing}.

\textbf{Challenge 2. A large test space.}
A test case for stateful bugs consists of a model and an invocation sequence.
For a fixed model, $k$ possible choices at each position give $k^n$
invocation sequences of length $n$. This makes exhaustive search infeasible
as sequences become longer. Finding sequences that trigger rarely observed
compiler behaviors is also difficult because compiler state is not directly
observable. For example, invocation $B$ may trigger a new compiler behavior
after $A$, but does not trigger one after $A'$. It is therefore difficult
to choose which sequences to extend or change, and random search may
repeatedly trigger common behaviors.

\textbf{Challenge 3. Oracle Unavailable.}
After running a test case, ground truth is often unavailable
for determining whether its results are correct~\cite{Xie2011MetamorphicML}. Traditional differential
oracles compare results from different modes, but an inconsistency alone does
not establish whether it is caused by a compiler bug or whether the bug
depends on compiler state. For example, an invocation may produce different
results with and without compilation. This comparison does not tell us
which result is incorrect. It also does not show whether compiler state
from earlier invocations caused the inconsistency.

To address these challenges, we formulate the testing problem and define
compiler state. The complete compiler state is not directly observable, so
we use invocation sequences to represent compiler state. We represent each
test case as
$\langle\textit{model},\allowbreak\textit{invocation sequence}\rangle$.
A fixed initial state and a fixed invocation sequence determine the
resulting compiler state. We can explore different compiler states by generating
invocation sequences and reach the same state again by replaying the same
sequence (addressing \textbf{Challenge 1}).
Based on this representation, we develop \rondo, a stateful testing method
for deep learning compilers. \rondo first generates test cases by mutating
invocation sequences or using an SMT solver.
It runs the invocation sequences and uses feedback to guide subsequent
test case generation toward rarely observed compiler behaviors
(addressing \textbf{Challenge 2}). It then compares results from invocations
without compiling the model, compiled invocations without compiler state
from earlier invocations, and compiled invocations with that state.
To check whether a bug depends on compilation or compiler
state, \rondo compares the same invocation with its inputs and program state
fixed. A stateful bug is detected when the result is incorrect with compiler
state from earlier invocations and correct after that state is reset
(addressing \textbf{Challenge 3}). Finally, it shortens
the invocation sequence while preserving the original bug.

\looseness=-1
We implement \rondo for PyTorch, JAX, TensorFlow, and MindSpore and evaluate it
using 24 common model architectures, including multilayer perceptrons (MLPs),
convolutional neural networks (CNNs), and Transformers. We detect 117 bugs,
including 35 stateful bugs, and submit 39 bug reports, including 18 for
stateful bugs. Developers confirm 35 reports, including 16 for stateful
bugs, and fix 18 bugs, including 11 stateful bugs. \rondo achieves
the highest behavior coverage and compiler Python source-line coverage
in the baseline comparison. Its bug checks achieve 99.26\% alarm precision.
It shortens 44 of 117 invocation sequences while preserving the original
bugs, reducing the lengths of these 44 sequences by 83.34\% on average.

Our work makes the following contributions.
\begin{itemize}
  \item \textbf{A new perspective on compiler testing.} We study stateful
  compiler bugs and define compiler state to describe their dependence on
  earlier invocations. We represent each test case as
  $\langle\textit{model},\allowbreak\textit{invocation sequence}\rangle$.
  The sequence establishes compiler state for later invocations. We detect
  stateful bugs by comparing the same invocation with compiler state
  retained and reset, while keeping its inputs and program state fixed.
  Our analysis of 26 publicly reported bug issues guides this design.

  \item \textbf{A novel stateful testing method.} We propose \rondo,
  which generates test cases through mutation or SMT-based synthesis and
  uses compilation feedback to guide generation toward rarely observed
  compiler behaviors. It checks results through structured comparison
  and semantic checks, then shortens invocation sequences while preserving
  the original bug and its reproduction conditions.

  \item \textbf{An empirical evaluation on production compilers.} We evaluate
  \rondo on four frameworks using 24 common model architectures. It achieves
  the highest behavior coverage and compiler Python source-line coverage
  among the evaluated methods. We detect 117 bugs, including 35 stateful
  bugs, and submit 39 bug reports, including 18 for stateful bugs.
  Developers confirm 35 reports, including 16 for stateful bugs, and fix
  18 bugs, including 11 stateful bugs.

  \item \textbf{Open-science artifacts.} To support future research and
  community reuse, we release \rondo's implementation, the results of bug
  discovery, and the experimental results reported in this paper~\cite{RondoArtifact}.
  We welcome community contributions to extend these artifacts.
\end{itemize}
\section{Key Observations}
\label{sec:empirical-study}
\label{sec:motivation}

We analyze public bug issues to understand stateful compiler bugs and
identify the testing requirements that guide the design of \rondo.
We first describe bug collection and analysis, then present three
observations on bug triggers, affected results, and the need to
characterize compiler state.

\subsection{Bug Collection and Analysis}
\label{sec:empirical-design}

The analysis consists of three steps: collecting candidate bug reports,
filtering them according to the inclusion criteria, and analyzing the
retained bug cases to identify testing requirements.

We collect candidate reports from the public GitHub issue trackers of
PyTorch, JAX, and TensorFlow and the public Gitee issue tracker of
MindSpore. These frameworks retain compiled artifacts across invocations.
We use search terms for compiler interfaces (\texttt{torch.compile},
\texttt{jit}, and \texttt{tf.function}), correctness symptoms
(\texttt{wrong} and \texttt{incorrect}), and compiler state
(\texttt{guard}, \texttt{cache}, \texttt{stale}, \texttt{second call},
and \texttt{recompile}), and inspect reports linked from relevant
discussions. We found 135 candidate reports after deduplicating the search
results by issue URL. We retain reports created on or after January 1, 2020
that concern compilation or compiled invocations and describe an incorrect output,
gradient, or observable program state. Each report must provide an
invocation sequence and evidence that prior compiled invocations affect a
later incorrect result, based on controlled runs or a mechanism
explanation. We exclude reports concerning only performance, compilation
exceptions, or crashes, and all security reports. The query below
summarizes the search scope. We translate it into each tracker's native
syntax.
\par\noindent\begin{minipage}{\linewidth}
\begin{lstlisting}[language=SQL,basicstyle=\ttfamily\small,
  keywordstyle=\bfseries,columns=fullflexible,keepspaces=true,
  breaklines=true,showstringspaces=false,numbers=none,frame=none,
  xleftmargin=2\parindent,xrightmargin=\parindent,
  aboveskip=6pt,belowskip=6pt]
repository IN (pytorch/pytorch, jax-ml/jax,
               tensorflow/tensorflow, mindspore/mindspore)
AND type = Issue
AND created >= 2020-01-01
AND text CONTAINS ANY ("torch.compile", "jit", "tf.function",
                      "wrong", "incorrect", "guard", "cache",
                      "stale", "second call", "recompile")
\end{lstlisting}
\end{minipage}\par

We filter each candidate report according to the inclusion criteria using its
description, reproduction code, developer discussions, and linked fixes
or regression tests. We record excluded reports in four categories.
\textbf{Irrelevant} includes reports outside the study's scope, including
help requests, performance issues, compilation exceptions or crashes
without an incorrect semantic result, and failures without evidence of
a role for prior invocations. This category also includes security
reports and reports discovered by \rondo. Excluding \rondo's reports
keeps the study independent of the tool's findings.
\textbf{Enhancement} includes feature requests or design proposals without
evidence of a correctness failure.
\textbf{Duplicate} includes explicit duplicates and related reports
grouped into the same bug case based on an identified omission or
failure mechanism.
\textbf{Not an Issue} includes reported inconsistencies explained by expected
behavior or invalid API usage, as established by developer explanations
or a correction of the original report. We retained 26 bug cases after
filtering (Table~\ref{tab:empirical-subjects}).

\begin{table}[t]
  \caption{Bug report filtering results.}
  \label{tab:empirical-subjects}
  \rondotablestyle
  \begin{tabular}{lrrrrrr}
    \toprule
    \rowcolor{TableHeader}
    Subject & Collected & \multicolumn{4}{c}{Filtered} & Studied \\
    \cmidrule(lr){3-6}
    \rowcolor{TableHeader}
    & & Irrelevant & Enhancement & Duplicate & Not an Issue & \\
    \midrule
    PyTorch & 77 & 38 & 1 & 12 & 5 & 21 \\
    JAX & 27 & 20 & 0 & 0 & 3 & 4 \\
    TensorFlow & 16 & 14 & 0 & 0 & 1 & 1 \\
    MindSpore & 15 & 15 & 0 & 0 & 0 & 0 \\
    \midrule
    \rowcolor{TableHighlight}
    \textbf{Total} & 135 & 87 & 1 & 12 & 9 & 26 \\
    \bottomrule
  \end{tabular}
\end{table}

We successfully reproduced all 26 bug cases on their affected versions.
We examine their invocation sequences, failure mechanisms, and effects
on outputs, gradients, and observable program state. We classify the
bug mechanisms using reproduction code, developer discussions, and fixes~\cite{Sun2016CompilerBugs}.

\subsection{Bug Triggers across Invocations}
\label{sec:empirical-existence}
\label{sec:empirical-history-evidence}

To detect stateful compiler bugs, we need to understand their triggering
conditions. We examine the invocation sequences of the 26 studied bug
cases. The three representative cases below illustrate how these
conditions arise across invocations.

\textbf{Bug Case 1: Changed function operation.}
PyTorch
issue~\href{https://github.com/pytorch/pytorch/issues/197811}{\#197811}
concerns two functions with the same code structure. One adds its inputs,
and the other multiplies them. The first invocation uses the addition
function with input values $6$ and $3$ and correctly returns $9$. The
compiler retains the compiled code from this invocation. The next
invocation uses the multiplication function with the same input
values. It should return $18$, but instead returns $9$.
Figure~\ref{fig:example} illustrates this case and shows that resetting
compiler state restores the correct result for the later invocation.

\begin{figure}[htbp]
  \centering
  \includegraphics[trim=27.75bp 0 33bp 0,clip,width=\linewidth]{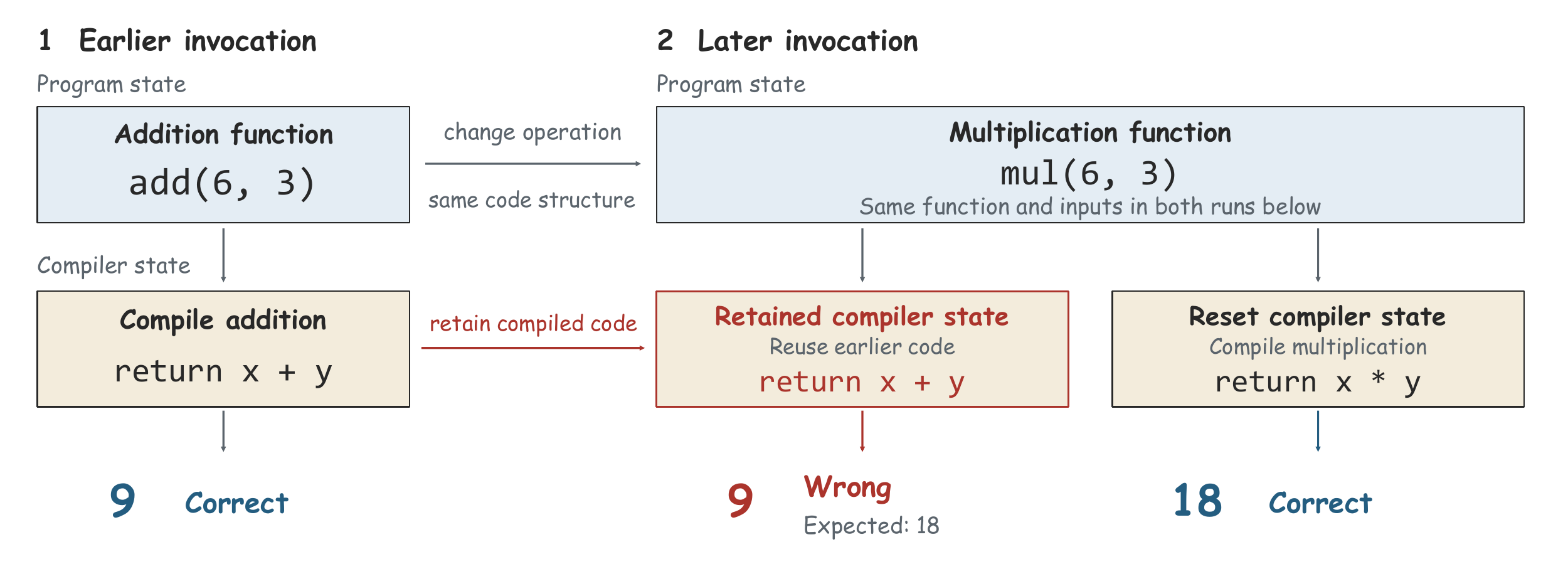}
  \caption{A stateful compiler bug across two invocations.}
  \label{fig:example}
  \Description{An earlier addition call with inputs six and three returns
  nine and leaves compiled addition code. A later multiplication call with
  the same inputs should return eighteen. With compiler state retained,
  it reuses the addition code and returns nine. Resetting compiler state
  while keeping the multiplication function and its inputs unchanged
  makes the same later invocation return the correct result eighteen.}
\end{figure}

\textbf{Bug Case 2: Changed gradient requirement.}
PyTorch
issue~\href{https://github.com/pytorch/pytorch/issues/162174}{\#162174}
concerns a change in whether a model's parameters require gradients.
Training uses gradients to determine how to adjust these numerical
parameters~\cite{Lee2020ADCorrectness,Lee2023MachineAD}. The first invocation does not require parameter gradients.
The compiler retains the compiled code from this invocation. Before the
next invocation, gradient computation is enabled for the same
parameters. The next invocation produces the correct output. It should
also produce parameter gradients with every element equal to $1$, but
instead produces no parameter gradients (\texttt{None}).

\textbf{Bug Case 3: Changed input alias relations.}
PyTorch
issue~\href{https://github.com/pytorch/pytorch/issues/139789}{\#139789}
concerns a change in whether two inputs share storage. The program doubles
the values stored in input \texttt{x} and then returns \texttt{x + y}.
The first invocation uses independent inputs, each containing four ones.
Doubling \texttt{x} leaves \texttt{y} unchanged, so the invocation correctly
returns $[3,3,3,3]$. The compiler retains the compiled code from this
invocation. The next invocation uses inputs with the same initial
values but with part of their storage shared. Doubling \texttt{x} also
changes the overlapping elements of \texttt{y}. It should return
$[4,4,3,3]$, but instead returns $[3,3,3,3]$.

In 17 of the 26 bug cases, the compiler reuses compiled code even though
the conditions required by that code no longer hold. In 9 cases, the
compiled code can still be used, but later invocations use earlier values
or structures that do not match the current program state. In both
situations, a later invocation incorrectly uses information retained by
the compiler from earlier invocations.
Triggering these bugs requires invocation sequences that first establish
compiler state and then change relevant program state before an invocation
that exposes the bug. The changes should cover parameter values and gradient
requirements, object attributes and bindings, input properties and alias
relations, function operations, modes, argument ordering, and
counters or random-generator state. These changes can leave object
identity, function code, or input values unchanged.

\observation{1}{Stateful compiler bugs require invocation sequences that
first establish compiler state and then expose an incorrect result in
a later invocation.}

\subsection{Impact of Stateful Bugs}
\label{sec:empirical-impact}

To detect stateful compiler bugs, we need to understand their observable
effects. We analyze these effects on outputs, gradients, and observable
program state to determine what testing must check.

\textbf{Outputs.}
Outputs are the results returned by a model or program. Incorrect outputs
can lead to incorrect predictions or task results. Model weights are
numerical parameters that determine a model's computation. PyTorch
issue~\href{https://github.com/pytorch/pytorch/issues/104984}{\#104984}
reports that a compiled model continues to use earlier weights after new
weights are loaded. PyTorch
issue~\href{https://github.com/pytorch/pytorch/issues/178365}{\#178365}
reports incorrect outputs after a function's default argument changes.
Resetting compiler state restores the correct outputs in the latter case.
The consequences can extend to the result of a complete task. JAX
issue~\href{https://github.com/jax-ml/jax/issues/31733}{\#31733}
shows that the same task returns the correct solution $1$ with an empty
persistent compilation cache, but returns the initial value $10$ with
an existing cache. Other cases incorrectly report whether two references
point to the same object or return an incorrect numeric data type, which
can change later program decisions or computations.

\textbf{Gradients.}
Training uses gradients to determine how to adjust a model's numerical
parameters. In PyTorch
issue~\href{https://github.com/pytorch/pytorch/issues/162174}{\#162174}
(Bug Case 2), the outputs are correct, but the parameter gradients are
missing (\texttt{None}). These missing gradients can prevent
the affected parameters from being updated as required, so correct
outputs alone do not guarantee correct training.

\textbf{Observable program state.}
Observable program state includes information a program keeps between
invocations and that we can inspect, such as the state of a random-number
generator.
A random-number generator updates this state as it produces random
results. PyTorch
issue~\href{https://github.com/pytorch/pytorch/issues/197085}{\#197085}
reports that repeated invocations of \texttt{shuffle} and \texttt{sample}
return the permutation or sample from the first invocation. There is also an
inconsistency between the random-generator state and that of the reference
run. These errors can affect both the current output and the random results of
subsequent invocations.

\textbf{Testing requirements.}
The checks must cover outputs, gradients, and observable program state.
For random operations,
the compared runs must start from the same random-generator state
and check both the result and the subsequent state.

\observation{2}{Detecting stateful compiler bugs requires checking
outputs, gradients, and observable program state. These bugs can affect
any of these results without raising an exception.}

\subsection{Characterizing Compiler State}
\label{sec:empirical-state-distinction}

To guide how testing represents state, we examine how stateful compiler
bugs differ from bugs exposed by an individual invocation. The studied
cases show that the invocation establishing the relevant compiler state
and the invocation exposing the bug are different.

A bug exposed by one invocation can be examined using that invocation's
model, inputs, and current program state. A stateful compiler bug also
depends on compiler state established by earlier invocations. In Bug
Case 2, the first invocation does not require parameter gradients and
produces the correct result. The compiler retains the compiled code from
this invocation. Before the next invocation, gradient computation is
enabled for the same parameters. The current program state now requires
parameter gradients, but the later invocation uses the earlier compiled
code and produces no parameter gradients (\texttt{None}). Resetting
compiler state restores the expected gradients while keeping the
parameter values, gradient requirements, and inputs unchanged. This case
shows that the invocation establishing compiler state can complete
correctly, even though a later invocation fails after a change to program
state. Program state determines the required result, while retained
compiler state can prevent that result from being produced. Describing
the current program state alone therefore does not capture all the
conditions needed to expose a stateful compiler bug.

Testing therefore needs to characterize the compiler state that affects
later invocations. This characterization should describe how the relevant
state was established and support reaching it again, so that tests can
explore different states and reproduce the conditions that exposed a bug.

\observation{3}{Stateful compiler testing needs to characterize compiler
state so that tests can establish the state needed to expose a bug and
reproduce its effect on a later invocation.}

\section{Related Work}
\label{sec:related}

DL compiler testing constructs models, tensor programs, or intermediate
representations (IRs), compiles and runs them, and applies an oracle
to identify crashes, incorrect results, or
inconsistencies~\cite{Liu2023NNSmith,Liu2022Tzer,Ma2023HirGen}.
Tests differ in the level at which they enter the compiler. Model- and
program-level tests exercise supported frontends and subsequent
transformations, while IR-level tests directly exercise compiler
optimizations and lowering. Methods can be grouped by how they construct
test cases: generation-based methods synthesize new models or programs
using operator constraints~\cite{Naziri2026Centaur}, language rules, compiler information, or
recorded operator usages, while mutation-based methods modify or combine
existing test cases.

\textbf{Generation-based methods.}
Generation-based methods construct new models or programs while satisfying
the constraints required for compilation~\cite{Zhang2017SPE}. NNSmith incrementally builds
valid computation graphs using operator specifications and tensor-type
propagation, and uses gradient-guided search to find inputs that avoid
exceptional floating-point values~\cite{Liu2023NNSmith}. NeuRI reduces the
manual specification effort by inferring operator constraints from valid
and invalid API traces and combining symbolic and concrete operators
during model generation~\cite{Liu2023NeuRI}. OPERA wraps operator
instances extracted from DL library tests as single-operator models to
test compiler frontends~\cite{Shen2025OPERA}. At the IR level, GenCoG uses
a domain-specific language to specify and solve operator type constraints
for generating diverse Relay graphs~\cite{Wang2023GenCoG}. HirGen generates
high-level IRs using
coverage criteria and IR language features~\cite{Ma2023HirGen}, while
MLIRSmith constructs MLIR program templates and instantiates them through
a context-sensitive grammar~\cite{Wang2023MLIRSmith}. WhiteFox uses large
language models (LLMs) to analyze optimization source code, extract
requirements for triggering the optimization, and generate programs that
satisfy those requirements. It uses tests that trigger optimizations as
feedback to guide subsequent generation~\cite{Yang2024WhiteFox}.

\textbf{Mutation-based methods.}
Mutation-based methods construct new tests by modifying or combining existing
models, programs, or IRs~\cite{Mu2024DevMuT,Zou2026Scalpel,Winterer2020SemanticFusion,Ge2026FlowMuT}. MT-DLComp creates model variants whose results
should agree for selected inputs and applies metamorphic
testing~\cite{Xiao2022MTDLComp}. TorchProbe uses semantics-preserving code
transformations to introduce dynamic Python features into test
programs~\cite{Su2023TorchProbe}. Tzer jointly mutates low-level IRs and
sequences of optimization passes under coverage feedback~\cite{Liu2022Tzer}.
SynthFuzz infers parameterized, context-dependent mutations from existing
MLIR tests and adapts them to new contexts~\cite{Limpanukorn2025SynthFuzz}.
AlignGuard extracts bug-triggering patterns from public reports and fixes,
uses LLMs to synthesize mutation rules, and applies them to existing
PyTorch compiler tests. Its patterns include repeated invocations with
changes to model state or input properties~\cite{Li2026AlignGuard}.
OATest extracts optimization patterns from developer tests and combines
them with seed computation graphs, repairing graph connections to satisfy
operator constraints~\cite{Shen2026OATest}. These methods reuse information
in existing tests to diversify program structures, the conditions under
which programs run, or compiler transformations.

\textbf{Test oracle.}
Differential testing compares compiled results with reference results~\cite{Mechtaev2018ReferenceRepair,Li2026DLLens} or
with results from different compiler configurations~\cite{Liu2023NNSmith,
Liu2022Tzer,Yang2024WhiteFox}. Metamorphic testing compares programs with a
declared relation between their results~\cite{Xiao2022MTDLComp,Mahathevan2026TENSURE,Xie2011MetamorphicML}.
DESIL applies rules to eliminate undefined behavior in MLIR tests and
selects lowering paths that allow the tests to run for differential
comparison~\cite{Suo2025DESIL}.

\textbf{Comparison with prior work.}
\rondo detects stateful compiler bugs by using invocation sequences to
represent and explore compiler state. First, compared with methods that generate or transform models,
programs, or IRs, it searches invocation sequences that establish compiler
state for later invocations. It mutates these sequences and declared state
operations, using compilation-behavior feedback to guide generation toward
rarely observed behaviors. Second, repeated-invocation tests can expose
failures after changes to model state or inputs, but confirming a stateful
compiler bug requires separating program state from compiler state.
\rondo compares the same invocation with compiler state retained and reset,
keeping its inputs and relevant program state unchanged. Third, it checks
outputs, gradients, and observable program state through structured
comparison and independent semantic checks. It then shortens the sequence
while preserving the original bug and, for a stateful bug, its dependence
on retained compiler state. Section~\ref{sec:methodology} presents the design.

\section{Methodology}
\label{sec:problem}
\label{sec:methodology}

This section presents the problem formulation and the testing method of
\rondo. We first specify the test representation and the bugs we aim to find, then
describe how \rondo generates test cases, checks their results,
and confirms and minimizes stateful witnesses.

\subsection{Problem Formulation}
\label{sec:problem-formulation}

Given a model $M$ and a fixed scenario, our goal is to find a
test case that exposes a stateful compiler bug. The
scenario fixes the framework build, backend, device, compiler options,
and initial conditions.

\begin{formulationstyle}
\begin{definition}[Compiler state]
\label{def:compiler-state}
Compiler state $\kappa$ is information retained by compilers across
invocations that can affect later invocations, excluding program state.
Examples include cached compiled artifacts, guards,
specializations, and profiling data. Program state $s$ contains
parameters, buffers, mutable inputs~\cite{Flanagan2009FastTrack}, alias relations~\cite{He2022Qilin,Wang2024DAInfer}, modes,
and other declared state. A configuration is $(s,\kappa)$.
\end{definition}
\end{formulationstyle}

The behavior required by program state must not depend on compiler state~\cite{Liu2022VerifiedScheduling,Liu2026VerifiedSparseTensor,Lu2025HieraSynth,Reinking2022Halide}.
The complete internal contents of compiler state depend on the compiler
and are not directly observable. Following the testing requirement in
Section~\ref{sec:empirical-state-distinction}, we represent compiler state by the
invocation sequence that establishes it from the scenario's fixed initial
state. We represent each test case as
\begin{equation}
  \mathcal{T}=\langle M,X\rangle,
  \qquad X=\langle a_1,\ldots,a_n\rangle.
  \label{eq:test-object}
\end{equation}
A test case consists of the model $M$ and the invocation sequence $X$.
Each $a_i$ specifies a model invocation with its inputs or a declared
state operation, such as updating a model parameter. The sequence records
these operations and their order.

\begin{formulationstyle}
\begin{definition}[Sequence representation of compiler state]
\label{def:history-representation}
Within the fixed scenario, replaying a legal prefix
$h=\langle a_1,\ldots,a_m\rangle$ from the initial configuration
$(s_0,\kappa_0)$ with compiler state retained determines the reached
configuration
\begin{equation}
  (s_h,\kappa_h)=\mathsf{Run}(s_0,\kappa_0,h).
  \label{eq:persistent-prefix}
\end{equation}
Thus, $\kappa_h$ is the compiler state represented by $h$ from
$(s_0,\kappa_0)$. The empty prefix leaves the initial configuration unchanged.
\end{definition}
\end{formulationstyle}

This representation allows \rondo to reach and reproduce compiler states
by generating and replaying invocation sequences. Different prefixes can
reach the same configuration. We seek a test case whose sequence is
$X=h\cdot a$, where the nonempty prefix $h$ establishes $(s_h,\kappa_h)$
for a later invocation $a$ that exposes a bug.
The same invocation $a$ can encounter different compiler states after
different prefixes. To confirm a stateful bug, $a$ must produce an
incorrect output, gradient, or observable program state with compiler
state retained, but produce the correct result after compiler state
is reset. The inputs to $a$ and relevant program state, including alias
relations, must remain unchanged between these two runs. This
comparison separates the effect of compiler state from the effect of
changes to program state.

The pair $(h,a)$ records the prefix and the invocation that exposes the
bug. We call this pair a stateful witness. \rondo must
generate a legal test case, detect the incorrect result, and establish
that retained compiler state is necessary. It then minimizes the sequence while
preserving the original bug and its witness class. The following
subsections describe how \rondo generates test cases, compares
observations, and refines candidates.

\subsection{Method Overview}
\label{sec:method-overview}

Figure~\ref{fig:pipeline} shows the three modules of \rondo, and
Algorithm~\ref{alg:rondo} summarizes the testing loop.
\emph{Test-space Exploration} first generates legal test cases
through mutation or, when compilation constraints are available,
SMT-based synthesis. It runs the generated test cases and uses
compilation-behavior feedback to guide an adaptive Markov chain Monte
Carlo (MCMC) search.
The \emph{Test Oracle} then runs selected test cases
under three modes: \emph{Without compilation}, \emph{Compiled without
state}, and \emph{Compiled with state}. Here, \emph{state} in the mode
names refers to compiler state from earlier invocations; all three modes
carry their own program state forward.
It compares their outputs, gradients, and observable program state using
representation-aware numeric rules and applies independent semantic
checks to identify candidate bugs. Finally, \emph{Stateful Candidate
Refinement} repeats the invocation that produced the incorrect result
with compiler state retained and reset, while keeping its inputs and
relevant program state unchanged. It determines the witness class and
minimizes the sequence while preserving the class, the invocation that
exposed the bug, and the recorded symptoms. Compilation behavior guides
test case generation, and this comparison determines whether the bug is stateful.

\begin{figure}[htbp]
  \centering
  \includegraphics[trim=21.75bp 0 21.75bp 0,clip,width=\linewidth]{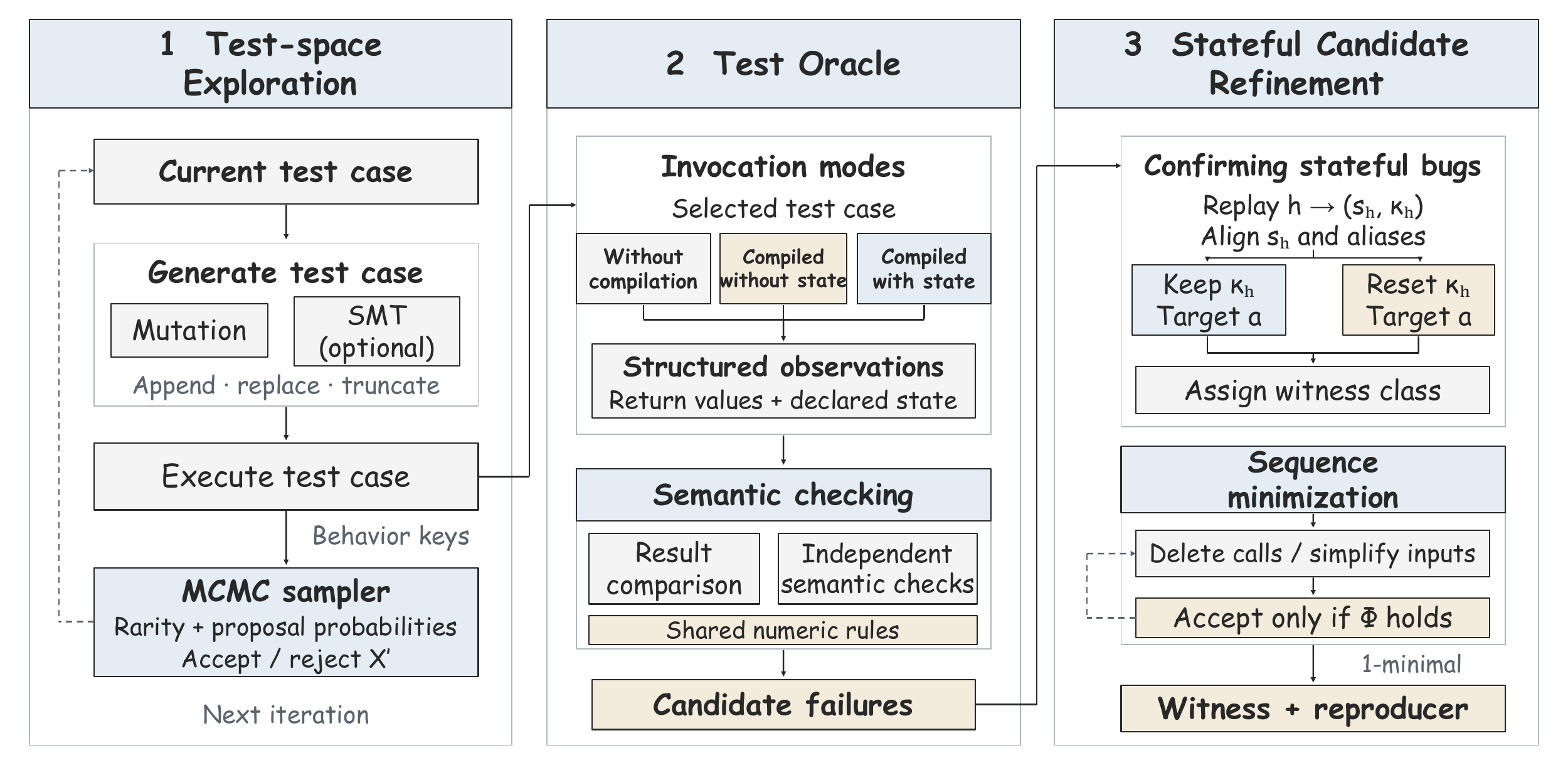}
  \caption{Overview of \rondo.}
  \label{fig:pipeline}
  \Description{Three panels show the complete testing pipeline. The first
  generates legal test cases through mutation or SMT-based synthesis,
  runs them, and feeds behavior keys to an MCMC
  sampler that guides subsequent test case generation. Selected test cases enter the
  second panel, where Without compilation, Compiled without state, and
  Compiled with state produce structured observations. Result comparison and independent semantic
  checks use shared numeric rules to identify candidate failures. In the
  third panel, the same invocation is repeated with compiler state retained
  and reset from aligned program state to determine the witness class. Sequence deletions and invocation
  simplifications are accepted only when the preservation condition holds,
  producing a witness and reproducer.}
\end{figure}

\begin{algorithm}[htbp]
  \caption{The \rondo algorithm}
  \label{alg:rondo}
  \small
  \begin{algorithmic}[1]
    \Require Model $M$, initial invocation sequence $X_0$, iteration budget $N$
    \Ensure Minimized bug witnesses $\mathcal B$
    \Function{Rondo}{$M,X_0,N$}
      \State $X \gets X_0$; $F \gets \emptyset$; $\mathcal B \gets \emptyset$
      \For{$t \gets 1$ \textbf{to} $N$}
        \State $X' \gets \Call{GenerateLegalSequence}{M,X,F}$ \Comment{1. Mutation or SMT synthesis}
        \State $K \gets \Call{CollectCompilationBehavior}{M,X'}$
        \State $\mathcal C \gets \Call{TestOracle}{M,X'}$ \Comment{2. Compare 3 modes; check semantics}
        \ForAll{$c \in \mathcal C$}
          \State $w \gets \Call{ConfirmAndClassify}{c}$ \Comment{3. Align program state; retain/reset compiler state}
          \If{$w.\mathit{class} \ne \textsf{undetermined}$}
            \State $\mathcal B \gets \mathcal B \cup \{\Call{Minimize}{w}\}$ \Comment{Preserve the bug and witness class}
          \EndIf
        \EndFor
        \State $X \gets \Call{MCMCSelect}{X,X',K,F}$ \Comment{Favor rare compilation behaviors}
        \State $F \gets \Call{UpdateFeedback}{F,K}$
      \EndFor
      \State \Return $\mathcal B$
    \EndFunction
  \end{algorithmic}
\end{algorithm}

The algorithm starts from $X_0$ and runs for $N$ iterations. Here, $F$
stores compilation feedback, $K$ contains the behavior keys observed for
the generated sequence, and $\mathcal C$ contains candidate bugs.

\subsection{Test-space Exploration}
\label{sec:design-space}

\rondo explores compiler states by generating invocation sequences.
It uses compilation-behavior feedback to guide test case generation.
In each iteration, \rondo uses mutation or SMT-based synthesis to generate
a new test case from the current sequence and runs it to collect
feedback. An adaptive MCMC sampler
uses this feedback to decide whether the new test case becomes the
current test case for the next iteration. We first describe test case
generation, then explain the compilation feedback and sampling procedure.

\subsubsection{Generating Test Cases}

Each scenario specifies the initial state. \rondo starts from an empty
sequence or a saved prefix. In each iteration, it generates a new test
case through feedback-guided mutation or SMT-based synthesis. Mutation
applies an operator that appends an invocation, replaces an existing invocation, or
removes invocations from the end of the sequence. Each
generated test case must satisfy the scenario's preconditions before it
is run.

\textbf{Feedback-guided mutation.}
Feedback-guided mutation aims to explore compiler states that arise when
inputs or program state change across invocations. The bug analysis in
Section~\ref{sec:empirical-existence} motivates varying input properties and
values, changing model state or gradient requirements, and testing later
invocations with earlier compiler state retained. We select 7 operators
for varying invocation inputs, 4 for state changes or gradient computation, and
2 for revisiting invocations or exploring alternative continuations.
\rondo uses these 13 mutation operators to generate new test cases from the
current invocation sequence, as summarized with examples in
Table~\ref{tab:mutation-operators}. The \textsf{shape}, \textsf{dtype},
\textsf{layout}, and \textsf{device} operators change the dimensions,
numeric type, storage arrangement, or device of a tensor,
a multidimensional array. The \textsf{resample} and \textsf{rescale}
operators change input values or their magnitudes, while
\textsf{special values} introduces NaN, infinity, very small subnormal
values, and positive or negative zero~\cite{Laguna2022BinFPE}. These mutations also cover
dimensions of size 0 and 1, values near numeric boundaries, and values
that affect program control flow~\cite{Laguna2023FPExceptions}. The \textsf{update} operator adds
parameter or buffer updates, \textsf{mode} adds switches between
training/evaluation or gradient-computation modes, and \textsf{restore}
restores saved program state. The \textsf{backward} operator adds gradient
computation for the parameters being tested. The \textsf{return} operator
repeats an earlier invocation with the same inputs and attributes, as in
$A\rightarrow B\rightarrow A$, to test it after the intervening invocations
with compiler state retained. The \textsf{truncate} operator removes
invocations from the end of the sequence so that later mutations can
generate different invocations after the remaining sequence. \rondo uses
compilation feedback to select operators, favoring those that previously
produced new behavior keys following the compilation behavior observed
at the end of the current sequence~\cite{Lee2023SeedAdaptive}.

\begin{table}[htbp]
  \caption{Mutation operators used by \rondo and illustrative examples.}
  \label{tab:mutation-operators}
  \rondotablestyle
  \renewcommand{\arraystretch}{1.14}
  \begin{tabularx}{\linewidth}{>{\raggedright\arraybackslash}p{0.15\linewidth}YY}
    \toprule
    \rowcolor{TableHeader}
    \textbf{Operator} & \textbf{Operation} & \textbf{Example change} \\
    \midrule
    Shape & Change tensor shape & $(1,4)\rightarrow(2,4)$ \\
    \rowcolor{TableStripe}
    Dtype & Change tensor dtype & \texttt{float32} $\rightarrow$ \texttt{float64} \\
    Layout & Change memory layout & contiguous $\rightarrow$ strided \\
    \rowcolor{TableStripe}
    Device & Change tensor device & CPU $\rightarrow$ GPU \\
    Resample & Resample input values & $x=3\rightarrow x=5$ \\
    \rowcolor{TableStripe}
    Rescale & Scale input magnitudes & $x=3\rightarrow x=3000$ \\
    Special values & Introduce special floating-point values & $+0.0\rightarrow-0.0$ \\
    \rowcolor{TableStripe}
    Update & Add a parameter or buffer update &
      $\langle A,B\rangle\rightarrow\langle A,B,U\rangle$; $U:w\leftarrow4$ \\
    Mode & Add a training or gradient mode switch &
      $\langle A,B\rangle\rightarrow\langle A,B,U\rangle$;
      $U:\mathrm{train}\rightarrow\mathrm{eval}$ \\
    \rowcolor{TableStripe}
    Restore & Restore saved program state &
      $\langle A,B\rangle\rightarrow\langle A,B,U\rangle$;
      $U:w\leftarrow w_{\mathrm{saved}}$ \\
    Backward & Add gradient computation &
      $\langle A,B\rangle\rightarrow\langle A,B,G\rangle$ \\
    \rowcolor{TableStripe}
    Return & Repeat an earlier invocation &
      $\langle A,B\rangle\rightarrow\langle A,B,A\rangle$ \\
    Truncate & Remove invocations from the end &
      $\langle A,B\rangle\rightarrow\langle A\rangle$ \\
    \bottomrule
  \end{tabularx}
  \par\smallskip
  \parbox{\linewidth}{\footnotesize
    Mutating an invocation's inputs produces a new invocation, which is
    appended to the sequence or replaces an existing invocation.
    In sequence examples, $A$ and $B$ are model invocations, $U$ is a
    declared state operation, $G$ is a backward pass, and $w$ is a model
    parameter.}
\end{table}

\textbf{SMT-based synthesis.}
Mutation may struggle to generate invocations whose inputs satisfy specific
compiler conditions. SMT-based synthesis targets these conditions directly.
\rondo also generates test cases through SMT-based synthesis when it can
obtain a guard or an input constraint from the framework. A guard is a condition
that checks whether a compiled artifact can be used for the current
invocation. \rondo selects an outcome of the condition that has not yet
been explored and asks the solver to generate a finite set of candidate
inputs that produce this outcome and satisfy the program's input
requirements~\cite{Lu2023Grisette,Chen2025FPConstraints}. For example, a condition may require a tensor dimension to
equal $1$ or be greater than $1$. \rondo selects a candidate from this set
and records its selection probability. It constructs an invocation with the
selected inputs, then appends it to the sequence or replaces an existing
invocation to generate a new test case. It checks the invocation for legality
before running the test case. If the solver cannot handle the constraints or finds no
solution, \rondo uses mutation to generate the test case.

\underline{\textbf{Example.}}
Consider an illustrative model $M(x)=wx$, where $x$ is a scalar input and
$w$ is a mutable parameter initially equal to $2$. Invocation $A$ uses
$x=3$, and invocation $C$ uses $x=5$. Starting from
$\langle A,C\rangle$, the \textsf{update} operator generates a new test
case by appending $U:w\leftarrow4$. A later round uses the
\textsf{return} operator to append $B$, which reissues $A$ with the same
input $x=3$. The generated sequence is $\langle A,C,U,B\rangle$.
All operations are legal, and
$B$ repeats $A$ after the parameter change with compiler state retained.
Applying \textsf{truncate} to this test case
can generate $\langle A,C,U\rangle$, so a later round can try a different
invocation after $U$. We use this model in the following examples.

\subsubsection{Guiding Test Case Generation}
\label{sec:design-explore}

A sequence that triggers no new compiler behavior may still trigger one
after further invocations are added, so the search should sometimes retain
such intermediate sequences.
\rondo uses an adaptive MCMC sampler to guide test case generation toward
rarely observed compilation behaviors. In each iteration, it runs a
new test case generated through mutation or SMT-based synthesis and uses
the compilation feedback to decide whether to continue from this test
case. Since the complete compiler state is not directly observable,
\rondo summarizes the observed compilation events and their relation to
preceding invocations and state changes as behavior keys. Let $K(X)$ contain
the keys observed when running sequence $X$, with repeated keys counted
separately, and let $c_t(k)$ count how often key $k$ was observed before
round $t$. \rondo gives higher scores to sequences with rarely observed
behaviors:
\begin{equation}
  S_t(X)=\sum_{k\in K(X)}\frac{1}{1+c_t(k)},
  \qquad
  \pi_t(X)\propto\exp\!\bigl(\beta S_t(X)\bigr).
  \label{eq:target}
\end{equation}
A behavior contributes less to the score as it is observed more often.
Here, $\pi_t$ is the sampling preference, and $\beta$ controls how strongly
the score affects this preference. The sampler compares the scores of the
current sequence $X$ and the new sequence $X'$ and accounts for the
probability of generating either sequence from the other. Let
$q_t(X'\mid X)$ denote the probability of generating $X'$ from $X$,
including the choice of mutation or SMT-based synthesis and the choices
made within that method. \rondo records these generation probabilities
and accepts $X'$ with probability
\begin{equation}
  \min\!\left\{1,
  \exp\!\bigl(\beta[S_t(X')-S_t(X)]\bigr)
  \frac{q_t(X\mid X')}{q_t(X'\mid X)}\right\}.
  \label{eq:acceptance-rule}
\end{equation}
If accepted, $X'$ becomes the starting point for generating the next test
case; otherwise, \rondo continues from $X$. The counts remain fixed during
this decision and are updated afterward. The generation probabilities
also change as the search proceeds. \rondo saves test cases that produce
new behavior keys as starting points for later searches~\cite{Wang2024DataCoverage}.

\underline{\textbf{Example.}}
Figure~\ref{fig:iteration} illustrates one iteration with the current
sequence $X=\langle A,C,U\rangle$.
The \textsf{return} operator generates a new test case with sequence
$X'=\langle A,C,U,B\rangle$; the reverse mutation uses
\textsf{truncate} to remove $B$. Suppose $A$ compiles an artifact, $C$
reuses it, and $B$ reuses it after the parameter update. The state
operation $U$ does not compile the model. Before this round, the
compilation and ordinary-reuse keys have each been observed nine times,
while reuse after $U$ has not been observed. Thus, $S_t(X)=0.2$ and
$S_t(X')=1.2$. Each key occurrence contributes to the score; another
ordinary reuse would add only $0.1$. For illustrative proposal probabilities
$q_t(X'\mid X)=1/4$ and $q_t(X\mid X')=1/16$, and $\beta=1$,
Equation~(\ref{eq:acceptance-rule}) gives an acceptance probability of
$\exp(1)/4\approx0.68$. A random draw of $0.5$ accepts $X'$.
The new reuse pattern guides search; the oracle must determine whether
the reuse is correct.

\begin{figure}[htbp]
  \centering
  \includegraphics[trim=24.75bp 0 24.75bp 0,clip,width=\linewidth]{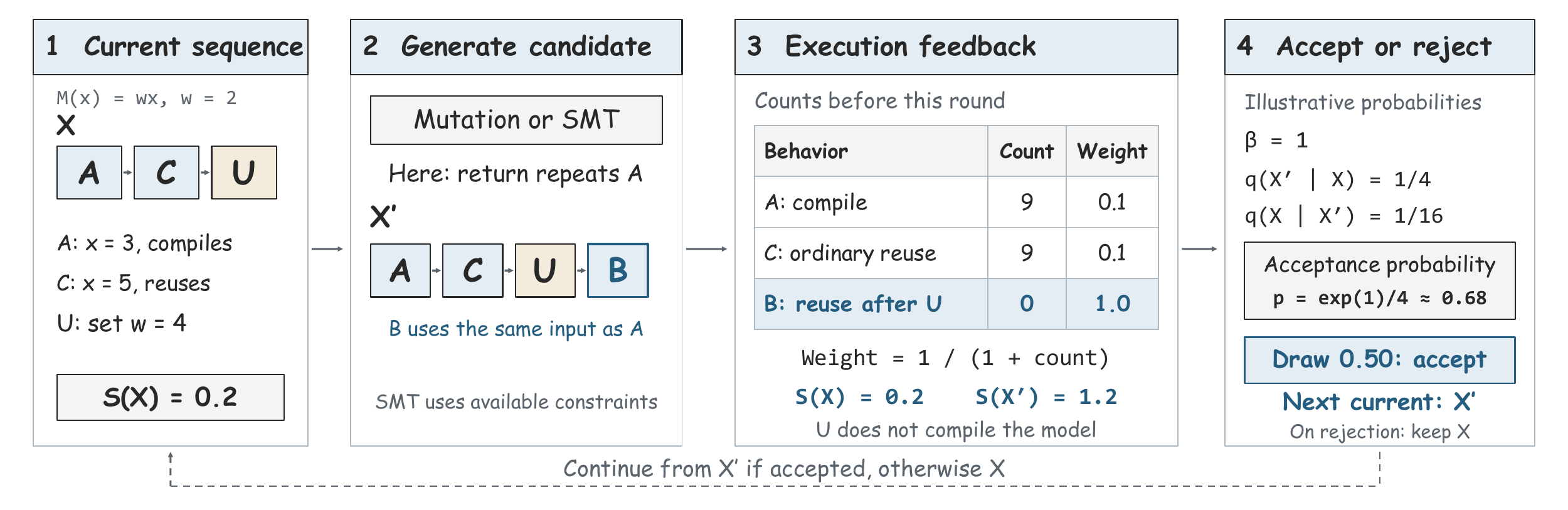}
  \caption{One iteration of test-space exploration.}
  \label{fig:iteration}
  \Description{Four stages show the current sequence, candidate
  generation, compilation-behavior feedback, and the acceptance decision. A return
  mutation appends invocation B to the sequence A, C, U. B repeats A after
  U updates the model parameter. Previously unobserved reuse after U raises
  the score from zero point two to one point two. With the illustrative
  proposal probabilities, the acceptance probability is approximately
  zero point six eight. A draw of zero point five accepts the candidate
  as the current sequence for the next iteration.}
\end{figure}

\subsection{Test Oracle}
\label{sec:design-paths}
\label{sec:design-observation}

\rondo checks selected test cases for incorrect outputs, gradients, and
observable program state. As shown in Figure~\ref{fig:oracle}, it runs
each test case under Without compilation, Compiled without state, and
Compiled with state, then
compares their structured observations and applies independent semantic
checks. A failed check produces a candidate bug for
\emph{Stateful Candidate Refinement}.

\begin{figure}[t]
  \centering
  \includegraphics[trim=21.75bp 0 21.75bp 0,clip,width=\linewidth]{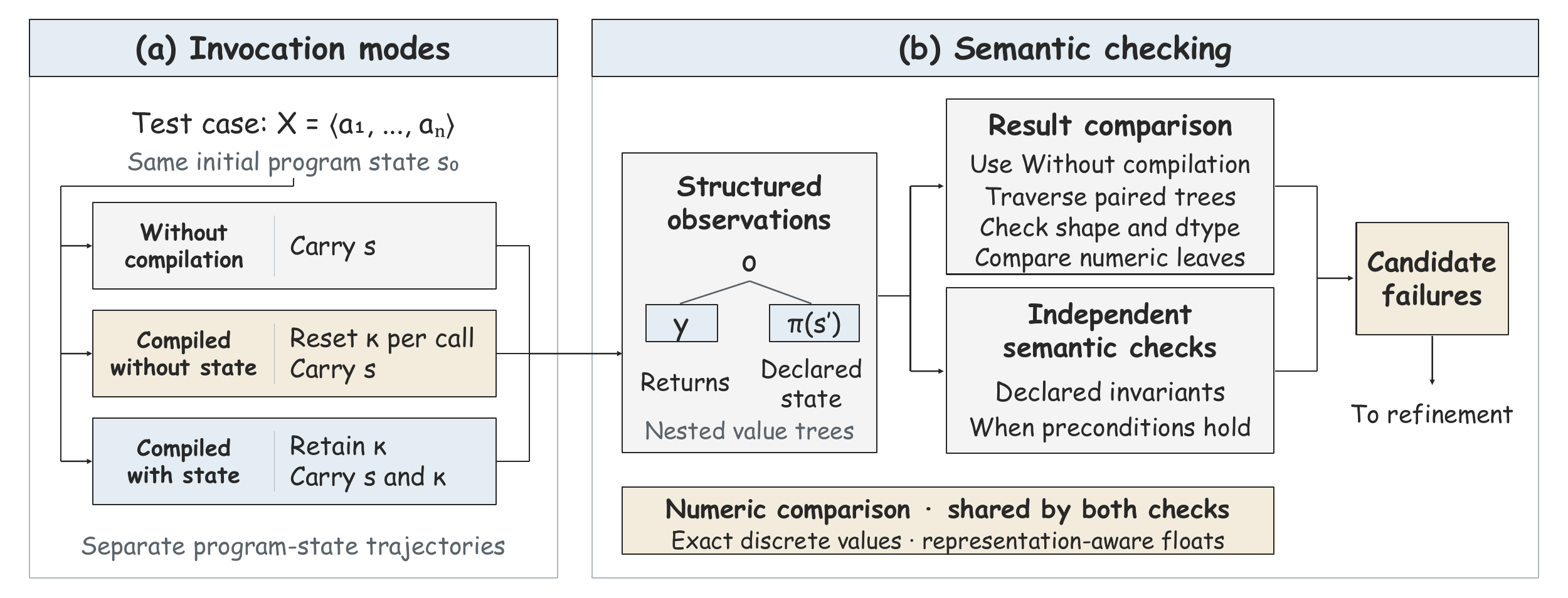}
  \caption{Test Oracle.}
  \label{fig:oracle}
  \Description{Two panels show invocation modes and semantic checking.
  Without compilation, Compiled without state, and Compiled with state run the same test case from
  the same initial program state, with separate program-state trajectories.
  Compiled without state resets compiler state before each invocation,
  while Compiled with state retains it. Their observations contain returned values and declared
  observable state. These observations enter two parallel paths: result
  comparison, which traverses paired trees and checks tensor attributes
  and numeric leaves, and independent semantic checks, which enforce
  declared invariants when their preconditions hold. Both paths share
  numeric comparison rules and can produce candidates for refinement.}
\end{figure}

\subsubsection{Invocation Modes}
\label{sec:design-policy-execution}
\label{sec:problem-state}

An inconsistency between compiled and uncompiled results does not show
whether retained compiler state is involved. Adding \emph{Compiled without
state} provides another comparison for identifying candidate stateful bugs.
The three modes run the same invocation sequence from the same
initial program state. They differ in whether they use compiled artifacts
and whether they retain compiler state.

All three modes start from program state $s_0$ and carry their own
program state forward. Changes to program state can change the required
behavior, while compiler state must not affect the behavior required by
the current program state. \textbf{Without compilation} invokes the model
without using the tested compilation path and provides observations for
comparison. \textbf{Compiled without state} resets
compiler state before each invocation, preserving the program state
reached by its preceding invocations. The reset boundary $b$ specifies
which compiler state \rondo clears or isolates. With fresh compiler
state denoted by $\kappa_0^b$, the reset is expressed as
\begin{equation}
  \mathsf{Reset}_{b}(s,\kappa)=(s,\kappa_0^b).
  \label{eq:history-erasure}
\end{equation}
The compiled invocation then starts from $(s,\kappa_0^b)$.
\textbf{Compiled with state} retains both program state and compiler state
throughout the sequence. Each invocation $a$ can reuse compiled artifacts
created by earlier invocations. The invocation updates both states,
written $\mathsf{Comp}_{M}(s,\kappa,a)=(s',\kappa',o)$, where $o$ is the
observation; the next invocation starts from $(s',\kappa')$.

\rondo implements \emph{Without compilation} using eager mode in PyTorch
and TensorFlow, non-jitted calls in JAX, and PyNative mode in MindSpore.
\emph{Compiled with state} keeps the same compiled callable throughout the sequence;
\emph{Compiled without state} rebuilds the callable and clears or isolates its compiler caches
at boundary $b$, using process isolation when required. Each mode
maintains its own program state, so the states reached before a later
invocation can differ. Result comparison identifies candidate bugs;
\rondo checks whether a bug is stateful in Section~\ref{sec:design-history-attribution}
by repeating the same invocation from aligned program state to determine
whether retained compiler state is necessary for the failure.

\subsubsection{Semantic Checking}
\label{sec:design-semantic-checking}

\looseness=-1
A model may return several values, produce gradients, or update observable
state. Checking these results requires a common description of what to
observe and what counts as correct.
\rondo checks observations under a correctness contract $\Gamma$, which
specifies the observable fields, correctness checks, and relevant program
state to align during refinement~\cite{Saan2024WitnessValidation}. We write $\mathsf{OK}_{\Gamma}(s,a,o)$
when observation $o$ satisfies these checks for invocation $a$ from
program state $s$.

\begin{formulationstyle}
\begin{definition}[Structured observation]
\label{def:structured-observation}
An invocation's structured observation is
$o=\langle y,\pi(s')\rangle$, where $y$ is its returned value tree and
$\pi(s')$ is the part of its resulting program state $s'$ declared
observable by $\Gamma$. The tree preserves containers, fields, nesting,
and the values or tensors at its leaves. Requested gradients are
recorded as returned values or declared state fields.
\end{definition}
\end{formulationstyle}

\looseness=-1
\textbf{Result comparison.}
Comparing values alone can miss errors in return structure, shape, or
dtype. A uniform numeric tolerance can also hide inconsistencies, such
as an incorrect zero replacing a small nonzero value.
At each invocation position $i$, \rondo compares observations
$o_i^{\mathrm{reset}}$ and $o_i^{\mathrm{keep}}$ from the two compiled modes
with observation $o_i^{\mathrm{uncompiled}}$ from Without compilation
under $\Gamma$. These observations include
outputs and declared state fields, such as parameters, buffers, and
mutable inputs. When gradients are requested, a backward pass computes
them; gradients are numerical derivatives used to adjust model parameters
during training. \rondo first checks that corresponding containers,
fields, nesting, tensor shapes~\cite{Lagouvardos2020TensorShape}, and dtypes agree, then compares their
values using representation-aware numeric rules~\cite{Li2024FTTN,Lee2023MixedPrecision,Fan2026ToolsCCBench}. Discrete values must
match exactly. Floating-point comparison checks NaN positions, infinity
signs, signed zeros, and whether one value is zero while the other is
nonzero. Normal finite nonzero values use dtype-specific relative and
absolute tolerances. For two subnormal values---nonzero values whose
magnitudes are below the smallest positive normal value of their
dtype---\rondo compares them on an integer scale derived from their raw
IEEE-754 encodings, using relative tolerance only. Floating-point values
are serialized as raw bits to
prevent changes to floating-point modes during a test run from affecting
the comparison. An inconsistency under these comparison rules produces a
candidate bug that records the invocation and its symptom signatures.
Each signature specifies the observation path, type of inconsistency, and
compared modes, allowing refinement to check that the original symptom
is preserved.

\label{sec:design-invariants}
\textbf{Independent semantic checks.}
\rondo also checks declared program properties to detect incorrect
results shared by all 3 modes. Its 6 checks cover required numerical
precision, stability of deterministic invocations, bit-pattern consistency
of identical scalars, mathematical identities across modes, consensus
among at least 3 implementations, and mathematical bounds on results.
Each check runs when its declared preconditions hold and uses the same
numeric rules as result comparison. Violations produce candidate bugs;
consensus identifies an outlier for further checking. Controlled replay
then determines whether a bug depends on compiler state.

\Needspace{8\baselineskip}
\underline{\textbf{Example.}}
For $\langle A,C,U,B\rangle$, suppose the compiler incorrectly retains
the parameter value $w=2$ in the artifact created by $A$. Without compilation,
$A$, $C$, and $B$ return $6$, $10$, and $12$, respectively.
Compiled without state, they return the same values because $B$ is compiled with the current
parameter $w=4$. Compiled with state, they return $6$, $10$, and $6$, because $B$
uses the earlier parameter value. At $B$, the observations
$\langle\text{returned value},w\rangle$ are $\langle12,4\rangle$
for Without compilation and Compiled without state and $\langle6,4\rangle$
for Compiled with state.
The returned-value inconsistency produces a candidate. The independent
semantic check $M(3)=4\cdot3=12$ also rejects the result from Compiled with state.

\subsection{Stateful Candidate Refinement}
\label{sec:design-refinement}
\label{sec:design-necessity}
\label{sec:design-reduction}

\rondo determines what a candidate needs for reproduction and
minimizes its invocation sequence while preserving those conditions.
It first checks whether a bug is stateful by repeating the invocation that produced
the incorrect result with compiler state retained and reset from aligned
program state. It then minimizes the sequence by removing invocations and
simplifying invocation inputs while preserving the invocation that exposed the bug,
its symptoms, witness class, and the relation defining that class.

\subsubsection{Confirming Stateful Bugs}
\label{sec:design-history-attribution}
\label{sec:problem-stateful}
\label{sec:problem-observation}
\label{sec:problem-equivalence}
\label{sec:problem-history}

For a candidate $h\cdot a$, $a$ is the invocation that produced the
incorrect result, and $h$ contains the preceding invocations. \rondo first
replays $h$ from the fixed initial configuration with compiler state
retained to reach $(s_h,\kappa_h)$, as defined by
Equation~(\ref{eq:persistent-prefix}).

Deleting the prefix can change both states, so \rondo repeats $a$ from
the reached program state $s_h$. It runs $a$ with compiler state $\kappa_h$
retained and with that state replaced by $\kappa_0^b$ through
$\mathsf{Reset}_b$. The inputs to $a$, relevant program state, and alias
relations are aligned. For random operations, this alignment also
includes the relevant random-generator state.

We write $\mathsf{Valid}_{\Gamma,b}(h,a)$ when the runs are legal
and stable, correctness judgments are determined, and the required
compiler-state reset and state alignment under $\mathsf{Align}_{\Gamma}$ are
achieved. If these conditions cannot be established, classification
remains \textsf{undetermined}.

\begin{formulationstyle}
\begin{definition}[Stateful witness]
\label{def:stateful-history}
Let $o_{\mathrm{keep}}$ and $o_{\mathrm{reset}}$ be the observations of $a$
with compiler state retained and reset from aligned program state
$s_h$. A pair $(h,a)$ is a stateful witness when
\begin{equation}
\begin{aligned}
  \mathsf{Stateful}_{\Gamma,b}(h,a) \iff{}&
  h\ne\langle\rangle \land \mathsf{Valid}_{\Gamma,b}(h,a) \\
  &\land\;\neg\mathsf{OK}_{\Gamma}(s_h,a,o_{\mathrm{keep}}) \\
  &\land\;\mathsf{OK}_{\Gamma}(s_h,a,o_{\mathrm{reset}}).
\end{aligned}
  \label{eq:dependent}
\end{equation}
\end{definition}
\end{formulationstyle}

\underline{\textbf{Example.}}
The prefix $h=\langle A,C,U\rangle$ leaves program state $w=4$ and
a compiled artifact that incorrectly uses $w=2$.
Figure~\ref{fig:method-cache-example} compares $B$ from this reached
program state. With compiler state retained, $B$ returns $6$.
Resetting compiler state while preserving $w=4$ and the input $x=3$
makes $B$ return the correct result $12$. This confirms that the bug is stateful. Deleting $h$ would also restore $w=2$, making $6$ the
correct result for a different program state. It therefore cannot
replace this comparison.

\begin{figure}[!h]
  \centering
\begin{lstlisting}[style=motivation,numbers=none,xleftmargin=0.7em,framexleftmargin=0.7em]
# Illustrative model; initial program state is w = 2
def model(x):
    return w * x
# h: invoke model(3), invoke model(5), then set w = 4
\end{lstlisting}
  \vspace{2pt}
  \begin{minipage}[t]{0.475\linewidth}
    \vspace{0pt}
\begin{lstlisting}[style=motivation,numbers=none,xleftmargin=0.7em,framexleftmargin=0.7em,backgroundcolor=\color{MethodGreenBg},rulecolor=\color{MethodGreenRule},escapeinside={(*@}{@*)}]
replay_prefix(h)  # w becomes 4
reset_compiler_state()  # keep w = 4
B = compiled_model(3)
# B = (*@\textcolor{MethodGreen}{\textbf{12, correct}}@*)
\end{lstlisting}
    \centering\textbf{(a) Reset compiler state}
  \end{minipage}\hfill
  \begin{minipage}[t]{0.475\linewidth}
    \vspace{0pt}
\begin{lstlisting}[style=motivation,numbers=none,xleftmargin=0.7em,framexleftmargin=0.7em,backgroundcolor=\color{MethodRedBg},rulecolor=\color{MethodRedRule},escapeinside={(*@}{@*)}]
replay_prefix(h)  # w becomes 4
# Retain A's artifact, which uses w = 2
B = compiled_model(3)
# B = (*@\textbf{6, wrong}@*)
\end{lstlisting}
    \centering\textbf{(b) Retained compiler state}
  \end{minipage}
  \caption{Confirming a stateful bug with the same program state
  $w=4$ and input $x=3$.}
  \label{fig:method-cache-example}
  \Description{An illustrative model multiplies its input by a mutable
  parameter. Both runs replay the same prefix, leaving the parameter
  equal to four. The left run resets compiler state while keeping
  the parameter and input unchanged, and returns twelve. The right
  run retains an artifact that incorrectly uses the earlier parameter
  value two, and returns 6.}
\end{figure}

Sequence minimization must preserve the conditions needed to reproduce
each bug. We therefore distinguish witnesses by whether reproduction
requires retained compiler state, program state established by the prefix,
or just the final invocation.
If $a$ still reproduces the original bug after compiler state is reset
with program state aligned, \rondo runs $a$ alone from the initial
program state with fresh compiler state. If that run reproduces the same symptom,
the witness is \emph{point-wise}. If it does not, the witness is
\emph{history-exposed}, because the program state established by the
prefix is needed, while its retained compiler state is not. Let
$\mathsf{Repro}_{\Gamma,R}(s,\kappa,a)$ mean that running $a$ reproduces
the original contract violation and symptom set $R$. These classes are
\label{sec:problem-classes}
\begin{align}
  \mathsf{PW}_{\Gamma,b,R}(a) \iff{}&
    \mathsf{Repro}_{\Gamma,R}(s_0,\kappa_0^b,a),
    \label{eq:pointwise}\\
  \mathsf{HE}_{\Gamma,b,R}(h,a) \iff{}&
    \mathsf{Repro}_{\Gamma,R}(s_h,\kappa_0^b,a)
    \land \neg\mathsf{PW}_{\Gamma,b,R}(a) \notag\\
    &\land \neg\mathsf{Stateful}_{\Gamma,b}(h,a).
    \label{eq:exposed}
\end{align}
These classes describe a witness's reproduction conditions. Another
witness of the same bug may belong to a different class.
A bug is classified as stateful if at least one of its witnesses
satisfies Definition~\ref{def:stateful-history}.

Because our tests represent compiler state through invocation sequences,
we show that the sequence-based and state-based conditions identify the
same stateful failure at the reached configuration.
The stateful condition can also be expressed at a reached
configuration $(s,\kappa)$. Under the same validity conditions, let
$o_{\kappa}$ and $o_0^b$ be the observations of $a$ with compiler state
retained and reset, respectively.
The state-based condition is
\begin{equation}
\begin{aligned}
  \mathsf{SB}_{\Gamma,b}(s,\kappa,a) \iff{}&
  \neg\mathsf{OK}_{\Gamma}(s,a,o_{\kappa}) \\
  &\land\;\mathsf{OK}_{\Gamma}(s,a,o_0^b).
\end{aligned}
  \label{eq:stateful-state}
\end{equation}
The following proposition relates the sequence-based and state-based
conditions at the configuration reached by the prefix.

\begin{formulationstyle}
\begin{proposition}[Sequence--state correspondence]
\label{prop:state-history}
For any nonempty legal prefix $h$ and invocation $a$ satisfying
$\mathsf{Valid}_{\Gamma,b}(h,a)$,
\begin{equation}
  \mathsf{Stateful}_{\Gamma,b}(h,a)
  \iff \mathsf{SB}_{\Gamma,b}(s_h,\kappa_h,a).
  \label{eq:state-history-equivalence}
\end{equation}
\end{proposition}
\begin{proof}
Replaying $h$ reaches $(s_h,\kappa_h)$ by
Equation~(\ref{eq:persistent-prefix}). Valid reset preserves $s_h$ and
replaces $\kappa_h$ with $\kappa_0^b$. Both conditions therefore compare
the same two runs of $a$ with compiler state retained and reset and make the same
correctness judgments.
\end{proof}
\end{formulationstyle}

\subsubsection{Sequence Minimization}

Sequence simplification can preserve an incorrect result while changing
a stateful witness into another witness class. \rondo therefore extends the failure predicate used
in delta debugging~\cite{Zeller2002Delta,Sun2018Perses} to preserve the witness class
and its defining relation. For a stateful witness, $a$
must remain incorrect with compiler state retained and correct after
compiler-state reset from aligned program state.

A shorter sequence can keep the same witness class while exposing a
different bug. Preserving the original failure therefore also requires
keeping its invocation, failing modes, and symptom signatures.
Let $\iota$ record the original witness class $c$, the identity $t$ of the
invocation that exposed the bug, failing-mode set $P$, and symptom
signatures $R$. Each mutation during minimization must satisfy three
checks. The runs must meet the
validity conditions in Section~\ref{sec:design-history-attribution}.
The witness must retain its class according to
Equation~(\ref{eq:dependent}), (\ref{eq:pointwise}), or
(\ref{eq:exposed}). It must also retain that invocation, failing
modes, and original symptom signatures, including the observation
paths, types of inconsistency, and semantic obligation of any independent
check. We denote the latter two checks by
$\mathsf{Class}_{c,\Gamma,b,\iota}$ and
$\mathsf{Preserve}_{\iota}$, respectively. A generated test case with
sequence $h'\cdot a'$ is accepted only when
\begin{equation}
\begin{aligned}
  \Phi_{\iota}(h',a') ={}&
  \mathsf{Valid}_{\Gamma,b}(h',a') \\
  &\land\;\mathsf{Class}_{c,\Gamma,b,\iota}(h',a') \\
  &\land\;\mathsf{Preserve}_{\iota}(h',a').
  \label{eq:refinement-condition}
\end{aligned}
\end{equation}

\rondo generates new test cases by mutating the current case: deleting
invocation blocks or individual invocations, or simplifying the inputs of
retained invocations. For each generated test case, it
replays the revised prefix and rechecks the witness class, including
runs of $a'$ from aligned program state. The new test case replaces
the current witness only when $\Phi_{\iota}$ holds. After exhausting single-invocation
deletions, the result is 1-minimal with respect to this condition and
that deletion operation.

\underline{\textbf{Example.}}
For $\langle A,C,U,B\rangle$, removing $C$ gives
$\langle A,U,B\rangle$. The same invocation $B$ still returns $6$
with compiler state retained and $12$ after compiler-state reset, so
the reduction preserves the original bug and its witness class.
Removing $A$ from this reduced sequence leaves $\langle U,B\rangle$:
$B$ is compiled for the first time with $w=4$ and correctly returns
$12$. Removing $U$ leaves $\langle A,B\rangle$, where $w$ remains
$2$ and $6$ is the correct result. Both deletions are rejected.
The final sequence retains the invocation that establishes the
compiled artifact, the parameter update, and the invocation that exposes
the bug.

\section{Experiment Design}
\label{sec:experiment-design}
\label{sec:setup}

\subsection{Research Questions}
\label{sec:evaluation-questions}

We evaluate \rondo through five research questions.

\begin{itemize}
  \item \textbf{RQ1: Bug Detection.} How effective is \rondo at detecting
  bugs in deep learning frameworks, and what types of bugs does it detect?
  \item \textbf{RQ2: Diversity and Coverage.} How does \rondo compare with
  the baselines in test diversity and compiler coverage?
  \item \textbf{RQ3: Component Ablation.} What effect does each exploration
  component have on test diversity and compiler coverage?
  \item \textbf{RQ4: Oracle Effectiveness.} How do the oracle configurations
  compare in true-positive detections and false alarms?
  \item \textbf{RQ5: Sequence Refinement.} How often does sequence refinement
  shorten an invocation sequence while preserving the original bug, and
  what is its cost?
\end{itemize}

\subsection{Baseline Methods}
\label{sec:baselines}

NNSmith, OATest, AlignGuard, and \textsc{RandomHistory} serve as baselines to
compare different testing strategies. We adapt their testing strategies for
this evaluation. For ablation, we compare full \rondo with four variants:
\emph{w/o Feedback}, which removes compilation-behavior feedback;
\emph{w/o SMT}, which disables SMT-based synthesis;
\emph{w/o Return}, which disables the \textsf{return} operator;
and \emph{w/o State Operations}, which disables state-changing actions.

\subsection{Measurements}
\label{sec:measurements}

\textbf{RQ1: Bug detection.}
We count each deduplicated bug once and report total and stateful bug
counts, together with their submission,
developer confirmation, and fix status. Oracle totals use test cases as
the counting unit.

\textbf{RQ2: Diversity and coverage.}
We use four metrics defined for this evaluation: Input, Path, Guard, and
History-Behavior Coverage (HBC). Input measures coverage of input
configurations, defined by dtype, rank, requested layout, and input-signature
region. Path measures coverage of operator signatures observed for different
input dtypes, ranks, requested layouts, and modes. Guard measures
coverage of compiler decisions (compilation, recompilation, or reuse) for
different input-signature regions, requested layouts, and modes.
HBC measures coverage of behavior transitions across invocations. A transition
records the state operation, preceding and current compilation decisions,
return status, mode, and whether the graph used in the invocation appears for
the first time. Programs and variants are distinguished in all four metrics.
The four metrics are computed as
$\mathrm{Input}=\frac{N_{\mathrm{Input,covered}}}{N_{\mathrm{Input,total}}}$,
$\mathrm{Path}=\frac{N_{\mathrm{Path,covered}}}{N_{\mathrm{Path,total}}}$,
$\mathrm{Guard}=\frac{N_{\mathrm{Guard,covered}}}{N_{\mathrm{Guard,total}}}$,
and
$\mathrm{HBC}=\frac{N_{\mathrm{HBC,covered}}}{N_{\mathrm{HBC,total}}}$.
Each numerator is the size of the set covered by the tests, and each
denominator is the size of the corresponding total set in the shared
evaluation space. The same total sets are used for all methods.
We use these metrics primarily to compare relative coverage between methods
within this evaluation space; their absolute values depend on the selected
total sets. We also report compiler Python source-line coverage as
$\frac{N_{\text{covered lines}}}{N_{\text{executable lines}}}$ in
\texttt{torch.\_dynamo}, \texttt{torch.\_inductor}, and
\texttt{torch.\_functorch}. We report cumulative coverage curves and
final values averaged across runs.

\textbf{RQ3: Component ablation.}
We measure final coverage and coverage growth, and report the
relative decrease from complete \rondo for each metric.

\textbf{RQ4: Oracle effectiveness.}
For each oracle configuration, we report cumulative true positives (TP), new TP, false positives
(FP), and total alarms. Alarm precision is
$\frac{\mathrm{TP}}{\mathrm{TP}+\mathrm{FP}}$, and the false discovery rate is
$\frac{\mathrm{FP}}{\mathrm{TP}+\mathrm{FP}}$. New TP counts test
cases detected at each added level that were missed by the preceding
configuration.

\textbf{RQ5: Sequence refinement.}
For each test case, we measure original and refined invocation
counts, the compression ratio
$1-\frac{N_{\text{refined calls}}}{N_{\text{original calls}}}$, the number of
oracle checks, and refinement time. Oracle checks determine whether a
mutated test case still exposes the original bug. Effective refinement requires a strictly shorter
sequence that preserves the original bug. Its success rate is
$\frac{N_{\text{effective refinements}}}{N_{\text{test cases}}}$.
We report means, standard deviations (SD), medians, and total oracle checks and
refinement time. Count medians are rounded to integers; compression-ratio
SD is expressed in percentage points. Cases stopped because of excessive time, inability to
shorten, or failure to preserve the original bug do not count as effective successes.

\subsection{Implementations}
\label{sec:implementation}

To fairly evaluate the performance of \rondo and baseline methods, we
implement them using the same framework versions, environment, and
configurations. \rondo is implemented in Python and tested on PyTorch
2.12.1, MindSpore 2.9.0, TensorFlow 2.21.0, and JAX
0.10.2. Specifically, \rondo consists of three modules:
\emph{Test-space Exploration}, \emph{Test Oracle}, and
\emph{Stateful Candidate Refinement}. The Test Oracle uses
Without compilation, Compiled without state, and Compiled with state modes and includes structured
observation comparison and independent semantic checks. The compilation
paths under test are \texttt{torch.compile} with Inductor, MindSpore graph
mode, \texttt{tf.function} with and without \texttt{jit\_compile}, and
\texttt{jax.jit} with XLA. Eager mode and MindSpore PyNative mode serve
as the corresponding paths without compilation. We select 24 common model
architectures, including multilayer perceptrons (MLPs), convolutional neural
networks (CNNs), and Transformers, as test programs. The experiments are
conducted on an
Intel Core i7-14700K machine with 32\,GB of RAM, Ubuntu 24.04.4 LTS,
and one 24\,GB NVIDIA GeForce RTX 4090 GPU. The Python version is
3.12.13, and PyTorch is built with CUDA 13.0. For bug discovery, we run
\rondo on the four frameworks using these model architectures. For coverage
comparison and component ablation, each method or configuration uses three
independent runs, and coverage is measured after 600 test cases.
The compared methods and ablation configurations use the same test programs,
seeds, settings, metric definitions, and reference protocol. For
oracle evaluation, we compare five levels of checking on the same 50{,}074
labeled test cases: crash-only checking, two-mode numerical comparison,
three-mode numerical comparison, three-mode structured comparison, and
Full Oracle. Each level includes the checks of the preceding configuration.
The two-mode configuration compares Without compilation and Compiled with state;
the three-mode configurations also compare Without compilation and Compiled without state.
Structured comparison additionally checks return structure, shape, dtype,
special floating-point values, and declared state fields. Full Oracle adds
semantic checks to three-mode structured comparison. For sequence refinement, statistics cover
test cases reproducing all 117 bugs, with the 44 effectively shortened
test cases reported separately.

\section{Results and Analysis}
\label{sec:results}

\subsection{RQ1: Bug Detection}
\label{sec:results-rq1}

\textbf{Detected bugs.}
\rondo detected 117 bugs across four frameworks, including
30 in PyTorch, 30 in TensorFlow, 17 in JAX, and
40 in MindSpore. We submitted 39 bug reports, including 15 for PyTorch,
9 for TensorFlow, 3 for JAX, and 12 for MindSpore. Of these reports, 18
concerned stateful bugs. Developers confirmed 35 reports, including 16 for
stateful bugs, and fixed 18 bugs, including 11 stateful bugs.
Table~\ref{tab:bug-status} summarizes these results.

\begin{table}[htbp]
  \caption{Bug detection and reporting results. Numbers in parentheses
  denote counts for bugs with at least one stateful witness.}
  \label{tab:bug-status}
  \rondotablestyle
  \begin{tabular}{lrrrr}
    \toprule
    \rowcolor{TableHeader}
    Framework & Bugs & Submitted & Confirmed & Fixed \\
    \midrule
    PyTorch & 30 (22) & 15 (14) & 13 (12) & 8 (8) \\
    TensorFlow & 30 (8) & 9 (1) & 7 (1) & 6 (1) \\
    JAX & 17 (2) & 3 (2) & 3 (2) & 3 (2) \\
    MindSpore & 40 (3) & 12 (1) & 12 (1) & 1 (0) \\
    \midrule
    \rowcolor{TableHighlight}
    \textbf{Total} & \textbf{117 (35)} & \textbf{39 (18)}
      & \textbf{35 (16)} & \textbf{18 (11)} \\
    \bottomrule
  \end{tabular}
\end{table}

\textbf{Stateful compiler bugs.}
Of the 117 bugs, 35 were stateful compiler bugs, including 22 in PyTorch,
8 in TensorFlow, 2 in JAX, and 3 in MindSpore. In these cases,
compiler state from earlier invocations caused an incorrect result in a
later invocation. For example,
PyTorch incorrectly shared compiled code between two functions with the
same code structure but different tensor operations. The second function
therefore performed the first function's computation. TensorFlow's
\texttt{tf.function} reused the same compiled graph for \texttt{+0.0} and
\texttt{-0.0}, although the function could distinguish them. The result
of the later call depended on which value appeared first. Changes to
methods, object attributes, or the default dtype also left compiled
PyTorch invocations using code for an earlier program state.

\textbf{Other bugs.}
The remaining 82 bugs included numerical errors, incorrect gradients,
dtype errors, crashes, and missing input checks. Examples included PyTorch
weight normalization producing incorrect outputs and NaN/Inf gradients,
JAX returning infinity for small upper-tail probabilities, and MindSpore
performing \texttt{float64} computations with \texttt{float32} precision.
Other cases caused PyTorch compilation to fail on symbolic expressions or
MindSpore calls to fail because of missing CPU kernels. Compiled
PyTorch also omitted the check that an input matrix is positive definite
in \texttt{torch.linalg.cholesky}.

\rqanswer{1}{\rondo detected 117 bugs across four
frameworks, including 35 stateful bugs. The remaining 82 bugs
included numerical errors, incorrect gradients, dtype errors, crashes,
and missing input checks. We submitted 39 bug reports, including 18 for
stateful bugs; developers confirmed 35 reports, including 16 for stateful
bugs, and fixed 18 bugs, including 11 stateful bugs.}

\subsection{RQ2: Diversity and Coverage}
\label{sec:results-rq2}

\textbf{Comparison with baselines.}
Table~\ref{tab:rq2-coverage} reports final behavior coverage and compiler
Python source-line coverage, while the top row of Figure~\ref{fig:coverage}
shows coverage growth.

\begin{table}[htbp]
  \caption{Behavior coverage and compiler Python source-line coverage.}
  \label{tab:rq2-coverage}
  \rondotablestyle
  \begin{tabular}{lrrrrrr}
    \toprule
    \rowcolor{TableHeader}
    & \multicolumn{4}{c}{Behavior coverage}
      & \multicolumn{2}{c}{Compiler coverage} \\
    \cmidrule(lr){2-5}\cmidrule(lr){6-7}
    \rowcolor{TableHeader}
    Method & Input & Path & Guard & HBC & Code (\%) & Covered lines \\
    \midrule
    NNSmith & 0.107 & 0.091 & 0.061 & 0.000 & 14.12 & 31{,}751 \\
    OATest & 0.069 & 0.056 & 0.035 & 0.000 & 14.15 & 31{,}820 \\
    AlignGuard & 0.061 & 0.059 & 0.082 & 0.023 & 14.07 & 31{,}627 \\
    \textsc{RandomHistory} & 0.109 & 0.106 & 0.154 & 0.117 & 14.17 & 31{,}864 \\
    \rowcolor{TableHighlight}
    \rondo & \textbf{0.356} & \textbf{0.370} & \textbf{0.384} & \textbf{0.242} & \textbf{14.94} & \textbf{33{,}591} \\
    \bottomrule
  \end{tabular}
\end{table}

\begin{figure}[htbp]
  \centering
  \includegraphics[trim=5bp 0 6.5bp 0,clip,width=\linewidth]{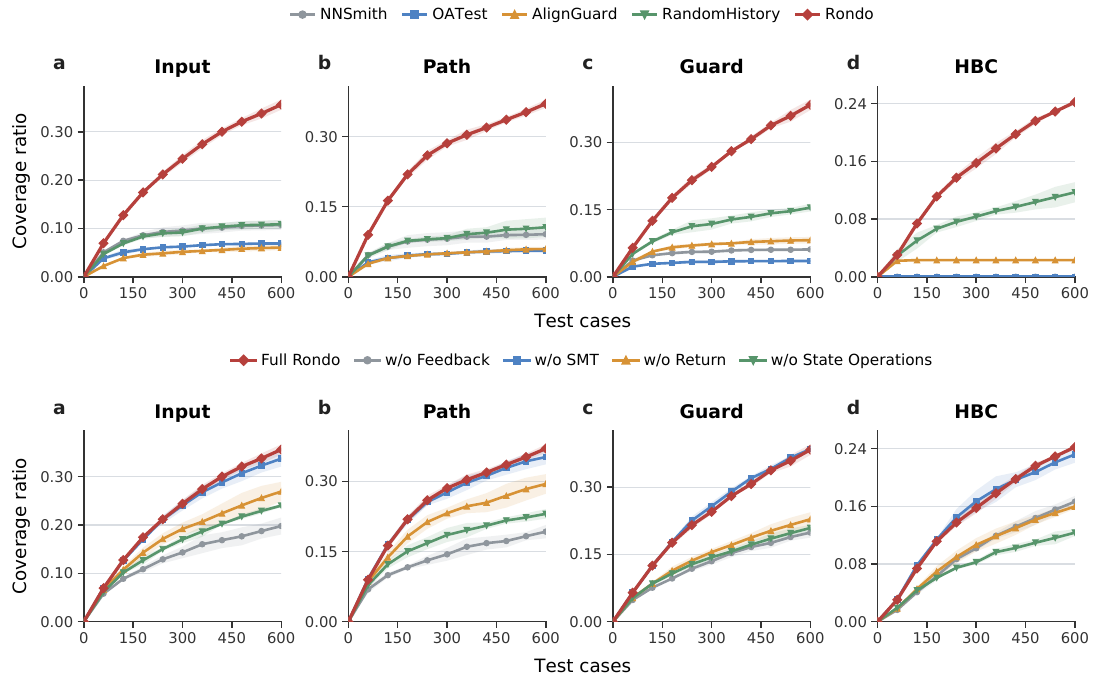}
  \caption{Coverage growth. Columns (a)--(d) show
  Input, Path, Guard, and HBC. The rows compare methods (top) and
  exploration-component configurations (bottom).}
  \label{fig:coverage}
  \Description{Eight line plots show cumulative behavior coverage over
  test cases. The top row compares NNSmith, OATest, AlignGuard,
  RandomHistory, and Rondo. The bottom row compares full Rondo with variants
  that remove feedback, SMT synthesis, the return operator, or state operations.
  The four columns show Input, Path, Guard, and History-Behavior Coverage.}
\end{figure}

\rondo achieves the highest coverage on all four behavior metrics.
Compared with \textsc{RandomHistory}, the strongest baseline on these
metrics, it increases Input and Path by 227.99\% and 249.58\%, and Guard
and HBC by 148.51\% and 106.66\%, respectively. The HBC gain shows that its
invocation sequences exercise more distinct behavior transitions.
\rondo also achieves the highest compiler Python source-line coverage:
14.94\%, compared with 14.17\% for \textsc{RandomHistory}. This is
an increase of 0.77 percentage points, corresponding to 1{,}727
additional covered lines. These results show that \rondo improves
both behavior coverage and compiler source-line coverage. The larger
relative gains in behavior coverage suggest that many additional sequences
exercise different invocation conditions within compiler code already
reached by the baselines. The two kinds of coverage therefore provide
complementary views of test diversity.

\rqanswer{2}{\rondo achieves the highest Input, Path, Guard, HBC, and
compiler Python source-line coverage among the evaluated methods. Its
HBC is 106.66\% higher than \textsc{RandomHistory}, the strongest baseline,
and its compiler Python source-line coverage is 14.94\%.}

\subsection{RQ3: Component Ablation}
\label{sec:results-rq3}

Table~\ref{tab:rq3-ablation} reports final coverage for the ablation
configurations, and the bottom row of Figure~\ref{fig:coverage} shows
coverage growth.

\begin{table}[htbp]
  \caption{Behavior coverage and compiler Python source-line coverage under exploration-component ablations.}
  \label{tab:rq3-ablation}
  \rondotablestyle
  \begin{tabular}{lrrrrrr}
    \toprule
    \rowcolor{TableHeader}
    & \multicolumn{4}{c}{Behavior coverage}
      & \multicolumn{2}{c}{Compiler coverage} \\
    \cmidrule(lr){2-5}\cmidrule(lr){6-7}
    \rowcolor{TableHeader}
    Configuration & Input & Path & Guard & HBC & Code (\%) & Covered lines \\
    \midrule
    \rowcolor{TableBaseline}
    Full \rondo & \textbf{0.356} & \textbf{0.370} & 0.384 & \textbf{0.242} & 14.94 & 33{,}591 \\
    \emph{w/o Feedback} & 0.198 & 0.193 & 0.199 & 0.166 & 14.79 & 33{,}239 \\
    \emph{w/o SMT} & 0.337 & 0.352 & \textbf{0.385} & 0.232 & \textbf{14.94} & \textbf{33{,}594} \\
    \emph{w/o Return} & 0.269 & 0.295 & 0.229 & 0.160 & 14.90 & 33{,}501 \\
    \emph{w/o State Operations} & 0.240 & 0.231 & 0.209 & 0.124 & 14.75 & 33{,}171 \\
    \bottomrule
  \end{tabular}
\end{table}

\textbf{Compilation-behavior feedback and SMT-based synthesis.}
Removing feedback causes the largest decreases in Input, Path, and Guard
among the four ablations, reducing them by 44.44\%, 48.00\%, and 48.19\%,
respectively. It also reduces HBC by 31.30\%. These results support using
compilation-behavior feedback to guide sequence search. Removing SMT-based
synthesis has a smaller effect: Input, Path, and HBC decrease by 5.31\%,
4.89\%, and 4.21\%, respectively. However, \emph{w/o SMT} slightly increases
Guard from 0.3837 to 0.3849 and covers three additional compiler Python
source lines; source-line coverage rounds to 14.94\% for both configurations. The measured
coverage gains from SMT-based synthesis are therefore in Input, Path,
and HBC.

\looseness=-1
\textbf{Return and state operations.}
Removing the \textsf{return} operator reduces Guard by 40.42\% and HBC
by 34.06\%, showing that revisiting earlier invocations contributes to
compiler-decision and behavior-transition coverage. Removing state
operations reduces HBC by 49.01\%, the largest HBC decrease among the
four ablations. It also reduces Input, Path, and Guard by 32.59\%,
37.71\%, and 45.50\%, respectively. These results show that changing
program state contributes to exploring behavior transitions across
invocations.

\rqanswer{3}{Removing feedback causes the largest decreases in Input,
Path, and Guard; removing state operations causes the largest HBC
decrease, at 49.01\%. The \textsf{return} operator improves all four behavior-coverage
metrics. SMT-based synthesis provides smaller gains in Input, Path,
and HBC. Compiler source-line coverage changes only slightly across
the configurations.}

\subsection{RQ4: Oracle Effectiveness}
\label{sec:results-rq4}

We compare the oracle configurations by the true-positive test cases they
detect and the false alarms they produce.

\begin{table}[htbp]
  \caption{Oracle detections and false discovery rates across five levels
  of checking. TP counts are cumulative.}
  \label{tab:rq4-oracle}
  \rondotablestyle
  \begin{tabularx}{\linewidth}{>{\raggedright\arraybackslash}p{0.35\linewidth}*{5}{Z}}
    \toprule
    \rowcolor{TableHeader}
    & \multicolumn{3}{c}{Detection results}
      & \multicolumn{2}{c}{Alarm reliability} \\
    \cmidrule(lr){2-4}\cmidrule(lr){5-6}
    \rowcolor{TableHeader}
    Configuration & TP & New TP & Alarms & FP & FDR (\%) \\
    \midrule
    Crash-only & 31 & 31 & 31 & 0 & 0.00 \\
    Two-mode numerical & 354 & 323 & 357 & 3 & 0.84 \\
    Three-mode numerical & 382 & 28 & 385 & 3 & 0.78 \\
    Three-mode structured & 487 & 105 & 490 & 3 & 0.61 \\
    \rowcolor{TableHighlight}
    \textbf{Full Oracle} & 539 & 52 & 543 & 4 & 0.74 \\
    \bottomrule
  \end{tabularx}
\end{table}

\textbf{Contributions of oracle checks.}
Table~\ref{tab:rq4-oracle} shows that true-positive detections increase
from 31 with crash-only checking to 539 with Full Oracle. Numerical comparison between
Without compilation and Compiled with state adds 323, raising the total to 354.
Adding Compiled without state
detects another 28, raising the total to 382. Structured comparison adds
105, raising the total to 487, a 27.49\% increase over three-mode
numerical comparison. Adding semantic checks detects another 52
true-positive test cases, bringing the Full Oracle total to 539.

\textbf{Alarm reliability.}
Full Oracle produces 543 alarms: 539 true positives and four false
positives. Its alarm precision is 99.26\%, and its false discovery rate
is 0.74\%, compared with 0.61\% for three-mode structured comparison.
Adding semantic checks detects 52 more true-positive test cases at the
cost of one additional false positive.

\rqanswer{4}{True-positive detections increase from 31 with crash-only
checking to 539 with Full Oracle. Semantic checks add 52 true-positive
test cases. Full Oracle produces four false positives among 543 alarms,
achieving 99.26\% alarm precision and a 0.74\% false discovery rate.}

\Needspace{8\baselineskip}
\subsection{RQ5: Sequence Refinement}
\label{sec:results-rq5}

\begin{table}[htbp]
  \caption{Sequence refinement statistics for all 117 test cases and the 44
  effectively shortened test cases.}
  \label{tab:rq5-refinement}
  \rondotablestyle
  \begin{tabularx}{\linewidth}{>{\raggedright\arraybackslash}p{0.28\linewidth}*{6}{Z}}
    \toprule
    \rowcolor{TableHeader}
    & \multicolumn{3}{c}{\textbf{All test cases (117)}}
      & \multicolumn{3}{c}{\textbf{Shortened test cases (44)}} \\
    \cmidrule(lr){2-4}\cmidrule(lr){5-7}
    \rowcolor{TableHeader}
    \textbf{Metric} & \textbf{Mean} & \textbf{SD} & \textbf{Median}
      & \textbf{Mean} & \textbf{SD} & \textbf{Median} \\
    \midrule
    Original invocations & 12.26 & 32.89 & 3 & 22.18 & 42.22 & 13 \\
    Refined invocations & 4.56 & 18.96 & 1 & 1.73 & 2.56 & 1 \\
    Oracle checks & 7.32 & 9.71 & 5 & 9.70 & 6.87 & 7 \\
    \midrule
    Compression ratio (\%) & 31.34 & 41.49 & 0.00 & 83.34 & 15.60 & 89.44 \\
    Refinement time (s) & 184.55 & 295.42 & 42.54 & 357.78 & 355.53 & 189.26 \\
    \bottomrule
  \end{tabularx}
\end{table}

\textbf{Sequence reduction and bug preservation.}
Table~\ref{tab:rq5-refinement} reports the refinement results.
Across all 117 test cases, the mean invocation count decreases from 12.26 to
4.56, and the median decreases from three to one. The mean compression
ratio is 31.34\%, with a median of zero. Effective refinement succeeds
for 44 test cases (37.61\%): their sequences are strictly shorter and
preserve the original bug. For these 44 test cases, the mean invocation count
decreases from 22.18 to 1.73, and the median decreases from about 13 to one.
The mean compression ratio is 83.34\%, with a median of 89.44\%.
The difference between the mean and median original sequence lengths
indicates that a few long sequences contribute substantially to the
total reduction in invocation count.

\looseness=-1
\textbf{Refinement cost.}
Across all 117 test cases, refinement uses 856 oracle checks and
21{,}592.8 seconds in total, averaging 7.32 oracle checks and 184.55 seconds
per test case. The 44 effectively shortened test cases account for 427 oracle
checks and 15{,}742.4 seconds, averaging 9.70 oracle checks and 357.78 seconds
per test case.

\rqanswer{5}{Sequence refinement shortens 44 of 117 test cases (37.61\%)
while preserving the original bugs. The mean compression ratio is 31.34\%
overall and 83.34\% for the 44 successful test cases. Across all test cases,
refinement averages 7.32 oracle checks and 184.55 seconds per test case.}

\section{Threats to Validity}
\label{sec:discussion}

\looseness=-1
\textbf{Internal Threats to Validity.} Errors in test runs, oracle checks, or
compiler-state reset can affect bug detection and the classification of stateful bugs.
To reduce these threats, we compare the results of Without compilation,
Compiled without state, and Compiled with state and use independent semantic
checks. We confirm stateful bugs
by replaying the same invocation with compiler state retained and reset
while aligning its inputs and relevant program state, and record the achieved
compiler-state reset boundary. Minimization
rechecks the original bug and its defining class relation after each reduction.
\textbf{External Threats to Validity.} The selected programs, frameworks,
and hardware may limit generalization. To reduce this threat, we evaluate
four frameworks using 24 common model architectures, including MLPs, CNNs,
and Transformers.
\textbf{Construct Threats to Validity.} Behavior metrics and oracle alarms
may not fully reflect compiler coverage or the number of underlying bugs.
We use common metric definitions, total sets, test case budgets, and
configurations for the compared methods, and report compiler Python
source-line coverage alongside behavior coverage. We also distinguish
detected bugs from oracle true-positive test cases and report the
contributions of oracle checks separately.

\section{Conclusion}
\label{sec:conclusion}

\looseness=-1
Single-invocation tests for deep learning compilers do not represent the
compiler state that can affect later invocations. We propose \rondo, a novel
stateful testing method
that uses invocation sequences to represent and explore hidden compiler
state. \rondo guides test case generation with compilation-behavior feedback
and checks results under Without compilation, Compiled without state,
and Compiled with state.
It confirms stateful bugs by comparing the same invocation with
compiler state retained and reset while preserving its inputs and
relevant program state. It then minimizes invocation sequences while
preserving the original bug and its defining class relation. Evaluation on four
deep learning frameworks detected 117 bugs, including 35 stateful bugs.
We submitted 39 bug reports, including 18 for stateful bugs. Developers
confirmed 35 reports, including 16 for stateful bugs, and fixed 18 bugs,
including 11 stateful bugs.

\section*{Data-Availability Statement}

\looseness=-1
The public artifact~\cite{RondoArtifact} contains \rondo's
implementation, bug-discovery results, and experimental results.
Reproduction requires the framework versions and device support described
in Section~\ref{sec:implementation}.

\bibliographystyle{ACM-Reference-Format}
\bibliography{references}

@article{Zeller2002Delta,
  author       = {Zeller, A. and Hildebrandt, R.},
  title        = {{Simplifying and Isolating Failure-Inducing Input}},
  journal      = {IEEE Transactions on Software Engineering},
  year         = {2002},
  volume       = {28},
  number       = {2},
  articleno    = {},
  pages        = {183--200},
  numpages     = {18},
  publisher    = {Institute of Electrical and Electronics Engineers (IEEE)},
  doi          = {10.1109/32.988498},
  url          = {}
}

@inproceedings{Pham2019CRADLE,
  author       = {Pham, Hung Viet and Lutellier, Thibaud and Qi, Weizhen and Tan, Lin},
  title        = {{CRADLE: Cross-Backend Validation to Detect and Localize Bugs in Deep Learning Libraries}},
  booktitle    = {Proceedings of the 41st IEEE/ACM International Conference on Software Engineering (ICSE 2019)},
  year         = {2019},
  series       = {},
  volume       = {},
  pages        = {1027--1038},
  numpages     = {12},
  publisher    = {Institute of Electrical and Electronics Engineers},
  address      = {Piscataway, NJ, USA},
  doi          = {10.1109/icse.2019.00107},
  url          = {}
}

@article{Liu2022Tzer,
  author       = {Liu, Jiawei and Wei, Yuxiang and Yang, Sen and Deng, Yinlin and Zhang, Lingming},
  title        = {{Coverage-Guided Tensor Compiler Fuzzing with Joint IR-Pass Mutation}},
  journal      = {Proceedings of the ACM on Programming Languages},
  year         = {2022},
  volume       = {6},
  number       = {OOPSLA1},
  articleno    = {73},
  pages        = {},
  numpages     = {26},
  publisher    = {Association for Computing Machinery},
  doi          = {10.1145/3527317},
  url          = {}
}

@article{Xiao2022MTDLComp,
  author       = {Xiao, Dongwei and LIU, Zhibo and Yuan, Yuanyuan and Pang, Qi and Wang, Shuai},
  title        = {{Metamorphic Testing of Deep Learning Compilers}},
  journal      = {Proceedings of the ACM on Measurement and Analysis of Computing Systems},
  year         = {2022},
  volume       = {6},
  number       = {1},
  articleno    = {},
  pages        = {1--28},
  numpages     = {28},
  publisher    = {Association for Computing Machinery},
  doi          = {10.1145/3508035},
  url          = {}
}

@inproceedings{Liu2023NNSmith,
  author       = {Liu, Jiawei and Lin, Jinkun and Ruffy, Fabian and Tan, Cheng and Li, Jinyang and Panda, Aurojit and Zhang, Lingming},
  title        = {{NNSmith: Generating Diverse and Valid Test Cases for Deep Learning Compilers}},
  booktitle    = {Proceedings of the 28th ACM International Conference on Architectural Support for Programming Languages and Operating Systems (ASPLOS 2023)},
  year         = {2023},
  series       = {},
  volume       = {2},
  pages        = {530--543},
  numpages     = {14},
  publisher    = {Association for Computing Machinery},
  address      = {New York, NY, USA},
  doi          = {10.1145/3575693.3575707},
  url          = {}
}

@inproceedings{Ma2023HirGen,
  author       = {Ma, Haoyang and Shen, Qingchao and Tian, Yongqiang and Chen, Junjie and Cheung, Shing-Chi},
  title        = {{Fuzzing Deep Learning Compilers with HirGen}},
  booktitle    = {Proceedings of the 32nd ACM SIGSOFT International Symposium on Software Testing and Analysis (ISSTA 2023)},
  year         = {2023},
  series       = {},
  volume       = {},
  pages        = {248--260},
  numpages     = {13},
  publisher    = {Association for Computing Machinery},
  address      = {New York, NY, USA},
  doi          = {10.1145/3597926.3598053},
  url          = {}
}

@inproceedings{Shen2026OATest,
  author       = {Shen, Qingchao and Wang, Zan and Ma, Haoyang and Tian, Yongqiang and Huang, Lili and Xiao, Zibo and Chen, Junjie and Cheung, Shing-Chi},
  title        = {{Optimization-Aware Test Generation for Deep Learning Compilers}},
  booktitle    = {Proceedings of the 48th IEEE/ACM International Conference on Software Engineering (ICSE 2026)},
  year         = {2026},
  series       = {},
  volume       = {},
  pages        = {3787--3799},
  numpages     = {13},
  publisher    = {Association for Computing Machinery},
  address      = {New York, NY, USA},
  doi          = {10.1145/3744916.3773216},
  url          = {}
}

@inproceedings{Su2023TorchProbe,
  author       = {Su, Qidong and Geng, Chuqin and Pekhimenko, Gennady and Si, Xujie},
  title        = {{TorchProbe: Fuzzing Dynamic Deep Learning Compilers}},
  booktitle    = {Programming Languages and Systems (APLAS 2023)},
  year         = {2023},
  series       = {Lecture Notes in Computer Science},
  volume       = {14405},
  pages        = {310--331},
  numpages     = {22},
  publisher    = {Springer Nature Singapore},
  address      = {Singapore},
  doi          = {10.1007/978-981-99-8311-7_15},
  url          = {}
}

@inproceedings{Ansel2024PyTorch2,
  author       = {Ansel, Jason and Yang, Edward and He, Horace and Gimelshein, Natalia and Jain, Animesh and Voznesensky, Michael and Bao, Bin and Bell, Peter and Berard, David and Burovski, Evgeni and Chauhan, Geeta and Chourdia, Anjali and Constable, Will and Desmaison, Alban and DeVito, Zachary and Ellison, Elias and Feng, Will and Gong, Jiong and Gschwind, Michael and Hirsh, Brian and Huang, Sherlock and Kalambarkar, Kshiteej and Kirsch, Laurent and Lazos, Michael and Lezcano, Mario and Liang, Yanbo and Liang, Jason and Lu, Yinghai and Luk, C. K. and Maher, Bert and Pan, Yunjie and Puhrsch, Christian and Reso, Matthias and Saroufim, Mark and Siraichi, Marcos Yukio and Suk, Helen and Zhang, Shunting and Suo, Michael and Tillet, Phil and Zhao, Xu and Wang, Eikan and Zhou, Keren and Zou, Richard and Wang, Xiaodong and Mathews, Ajit and Wen, William and Chanan, Gregory and Wu, Peng and Chintala, Soumith},
  title        = {{PyTorch 2: Faster Machine Learning Through Dynamic Python Bytecode Transformation and Graph Compilation}},
  booktitle    = {Proceedings of the 29th ACM International Conference on Architectural Support for Programming Languages and Operating Systems (ASPLOS 2024)},
  year         = {2024},
  series       = {},
  volume       = {2},
  pages        = {929--947},
  numpages     = {19},
  publisher    = {Association for Computing Machinery},
  address      = {New York, NY, USA},
  doi          = {10.1145/3620665.3640366},
  url          = {}
}

@misc{Li2026AlignGuard,
  author       = {Li, Meiziniu and Li, Dongze and Liu, Jianmeng and Cheung, Shing-Chi},
  title        = {{Demystifying the Silence of Correctness Bugs in PyTorch Compiler}},
  year         = {2026},
  howpublished = {arXiv preprint},
  eprint       = {2604.08720},
  archivePrefix= {arXiv},
  doi          = {10.48550/arXiv.2604.08720},
  url          = {https://arxiv.org/abs/2604.08720},
  note         = {Version 1}
}

@inproceedings{Liu2023NeuRI,
  author       = {Liu, Jiawei and Peng, Jinjun and Wang, Yuyao and Zhang, Lingming},
  title        = {{NeuRI: Diversifying DNN Generation via Inductive Rule Inference}},
  booktitle    = {Proceedings of the 31st ACM Joint European Software Engineering Conference and Symposium on the Foundations of Software Engineering (ESEC/FSE 2023)},
  year         = {2023},
  series       = {},
  volume       = {},
  pages        = {657--669},
  numpages     = {13},
  publisher    = {Association for Computing Machinery},
  address      = {New York, NY, USA},
  doi          = {10.1145/3611643.3616337},
  url          = {}
}

@inproceedings{Wang2023GenCoG,
  author       = {Wang, Zihan and Nie, Pengbo and Miao, Xinyuan and Chen, Yuting and Wan, Chengcheng and Bu, Lei and Zhao, Jianjun},
  title        = {{GenCoG: A DSL-Based Approach to Generating Computation Graphs for TVM Testing}},
  booktitle    = {Proceedings of the 32nd ACM SIGSOFT International Symposium on Software Testing and Analysis (ISSTA 2023)},
  year         = {2023},
  series       = {},
  volume       = {},
  pages        = {904--916},
  numpages     = {13},
  publisher    = {Association for Computing Machinery},
  address      = {New York, NY, USA},
  doi          = {10.1145/3597926.3598105},
  url          = {}
}

@article{Yang2024WhiteFox,
  author       = {Yang, Chenyuan and Deng, Yinlin and Lu, Runyu and Yao, Jiayi and Liu, Jiawei and Jabbarvand, Reyhaneh and Zhang, Lingming},
  title        = {{WhiteFox: White-Box Compiler Fuzzing Empowered by Large Language Models}},
  journal      = {Proceedings of the ACM on Programming Languages},
  year         = {2024},
  volume       = {8},
  number       = {OOPSLA2},
  articleno    = {296},
  pages        = {},
  numpages     = {27},
  publisher    = {Association for Computing Machinery},
  doi          = {10.1145/3689736},
  url          = {}
}

@inproceedings{Wang2023MLIRSmith,
  author       = {Wang, Haoyu and Chen, Junjie and Xie, Chuyue and Liu, Shuang and Wang, Zan and Shen, Qingchao and Zhao, Yingquan},
  title        = {{MLIRSmith: Random Program Generation for Fuzzing MLIR Compiler Infrastructure}},
  booktitle    = {Proceedings of the 38th IEEE/ACM International Conference on Automated Software Engineering (ASE 2023)},
  year         = {2023},
  series       = {},
  volume       = {},
  pages        = {1555--1566},
  numpages     = {12},
  publisher    = {Institute of Electrical and Electronics Engineers},
  address      = {Piscataway, NJ, USA},
  doi          = {10.1109/ase56229.2023.00120},
  url          = {}
}

@inproceedings{Limpanukorn2025SynthFuzz,
  author       = {Limpanukorn, Ben and Wang, Jiyuan and Kang, Hong Jin and Zhou, Zitong and Kim, Miryung},
  title        = {{Fuzzing MLIR Compilers with Custom Mutation Synthesis}},
  booktitle    = {Proceedings of the 47th IEEE/ACM International Conference on Software Engineering (ICSE 2025)},
  year         = {2025},
  series       = {},
  volume       = {},
  pages        = {217--229},
  numpages     = {13},
  publisher    = {Institute of Electrical and Electronics Engineers},
  address      = {Piscataway, NJ, USA},
  doi          = {10.1109/icse55347.2025.00037},
  url          = {}
}

@inproceedings{Shen2025OPERA,
  author       = {Shen, Qingchao and Tian, Yongqiang and Ma, Haoyang and Chen, Junjie and Huang, Lili and Fu, Ruifeng and Cheung, Shing-Chi and Wang, Zan},
  title        = {{A Tale of Two DL Cities: When Library Tests Meet Compiler}},
  booktitle    = {Proceedings of the 47th IEEE/ACM International Conference on Software Engineering (ICSE 2025)},
  year         = {2025},
  series       = {},
  volume       = {},
  pages        = {2201--2212},
  numpages     = {12},
  publisher    = {Institute of Electrical and Electronics Engineers},
  address      = {Piscataway, NJ, USA},
  doi          = {10.1109/icse55347.2025.00025},
  url          = {}
}

@article{Suo2025DESIL,
  author       = {Suo, Chenyao and Wang, Jianrong and Wang, Yongjia and Jiang, Jiajun and Shen, Qingchao and Chen, Junjie},
  title        = {{DESIL: Detecting Silent Bugs in MLIR Compiler Infrastructure}},
  journal      = {Proceedings of the ACM on Programming Languages},
  year         = {2025},
  volume       = {9},
  number       = {OOPSLA2},
  articleno    = {383},
  pages        = {},
  numpages     = {27},
  publisher    = {Association for Computing Machinery},
  doi          = {10.1145/3763161},
  url          = {}
}

@article{Arora2025TensorRight,
  author       = {Arora, Jai and Lu, Sirui and Jain, Devansh and Xu, Tianfan and Houshmand, Farzin and Phothilimthana, Phitchaya Mangpo and Lesani, Mohsen and Narayanan, Praveen and Murthy, Karthik Srinivasa and Bodik, Rastislav and Sabne, Amit and Mendis, Charith},
  title        = {{TensorRight: Automated Verification of Tensor Graph Rewrites}},
  journal      = {Proceedings of the ACM on Programming Languages},
  year         = {2025},
  volume       = {9},
  number       = {POPL},
  articleno    = {29},
  pages        = {},
  numpages     = {32},
  publisher    = {Association for Computing Machinery},
  doi          = {10.1145/3704865},
  url          = {}
}

@inproceedings{Mahathevan2026TENSURE,
  author       = {Mahathevan, Kabilan and Zhang, Yining and Gulzar, Muhammad Ali and Sundararajah, Kirshanthan},
  title        = {{TenSure: Fuzzing Sparse Tensor Compilers (Registered Report)}},
  booktitle    = {Proceedings of the International Fuzzing Workshop (Fuzzing 2026)},
  year         = {2026},
  series       = {},
  volume       = {},
  pages        = {},
  numpages     = {10},
  publisher    = {Internet Society},
  address      = {San Diego, CA, USA},
  doi          = {10.14722/fuzzing.2026.23006},
  url          = {}
}

@inproceedings{Naziri2026Centaur,
  author       = {Naziri, M M Abid and Kim, Shinhae and Qin, Feiran (Alex) and Dutta, Saikat and d'Amorim, Marcelo},
  title        = {{Testing Deep Learning Libraries via Neurosymbolic Constraint Learning}},
  booktitle    = {Proceedings of the 48th IEEE/ACM International Conference on Software Engineering (ICSE 2026)},
  year         = {2026},
  series       = {},
  volume       = {},
  pages        = {3586--3598},
  numpages     = {13},
  publisher    = {Association for Computing Machinery},
  address      = {New York, NY, USA},
  doi          = {10.1145/3744916.3787800},
  url          = {}
}

@inproceedings{Mu2024DevMuT,
  author       = {Mu, Yanzhou and Zhai, Juan and Fang, Chunrong and Chen, Xiang and Cao, Zhixiang and Yang, Peiran and Zou, Yinglong and Zheng, Tao and Chen, Zhenyu},
  title        = {{DevMuT: Testing Deep Learning Framework via Developer Expertise-Based Mutation}},
  booktitle    = {Proceedings of the 39th IEEE/ACM International Conference on Automated Software Engineering (ASE 2024)},
  year         = {2024},
  series       = {},
  volume       = {},
  pages        = {1533--1544},
  numpages     = {12},
  publisher    = {Association for Computing Machinery},
  address      = {New York, NY, USA},
  doi          = {10.1145/3691620.3695523},
  url          = {}
}

@inproceedings{Zou2026Scalpel,
  author       = {Zou, Yinglong and Zhai, Juan and Fang, Chunrong and Guo, An and Liu, Jiawei and Chen, Zhenyu},
  title        = {{Scalpel: Automotive Deep Learning Framework Testing via Assembling Model Components}},
  booktitle    = {Proceedings of the 48th IEEE/ACM International Conference on Software Engineering (ICSE 2026)},
  year         = {2026},
  series       = {},
  volume       = {},
  pages        = {3298--3310},
  numpages     = {13},
  publisher    = {Association for Computing Machinery},
  address      = {New York, NY, USA},
  doi          = {10.1145/3744916.3773242},
  url          = {}
}

@inproceedings{Zhang2017SPE,
  author       = {Zhang, Qirun and Sun, Chengnian and Su, Zhendong},
  title        = {{Skeletal Program Enumeration for Rigorous Compiler Testing}},
  booktitle    = {Proceedings of the 38th ACM SIGPLAN Conference on Programming Language Design and Implementation (PLDI 2017)},
  year         = {2017},
  series       = {},
  volume       = {},
  pages        = {347--361},
  numpages     = {15},
  publisher    = {Association for Computing Machinery},
  address      = {New York, NY, USA},
  doi          = {10.1145/3062341.3062379},
  url          = {}
}

@inproceedings{Sun2018Perses,
  author       = {Sun, Chengnian and Li, Yuanbo and Zhang, Qirun and Gu, Tianxiao and Su, Zhendong},
  title        = {{Perses: Syntax-Guided Program Reduction}},
  booktitle    = {Proceedings of the 40th IEEE/ACM International Conference on Software Engineering (ICSE 2018)},
  year         = {2018},
  series       = {},
  volume       = {},
  pages        = {361--371},
  numpages     = {11},
  publisher    = {Association for Computing Machinery},
  address      = {New York, NY, USA},
  doi          = {10.1145/3180155.3180236},
  url          = {}
}

@inproceedings{Sun2016CompilerBugs,
  author       = {Sun, Chengnian and Le, Vu and Zhang, Qirun and Su, Zhendong},
  title        = {{Toward Understanding Compiler Bugs in GCC and LLVM}},
  booktitle    = {Proceedings of the 25th International Symposium on Software Testing and Analysis (ISSTA 2016)},
  year         = {2016},
  series       = {},
  volume       = {},
  pages        = {294--305},
  numpages     = {12},
  publisher    = {Association for Computing Machinery},
  address      = {New York, NY, USA},
  doi          = {10.1145/2931037.2931074},
  url          = {}
}

@article{Liu2022VerifiedScheduling,
  author       = {Liu, Amanda and Bernstein, Gilbert Louis and Chlipala, Adam and Ragan-Kelley, Jonathan},
  title        = {{Verified Tensor-Program Optimization via High-Level Scheduling Rewrites}},
  journal      = {Proceedings of the ACM on Programming Languages},
  year         = {2022},
  volume       = {6},
  number       = {POPL},
  articleno    = {55},
  pages        = {},
  numpages     = {28},
  publisher    = {Association for Computing Machinery},
  doi          = {10.1145/3498717},
  url          = {}
}

@article{Liu2024FunctionalTensorCompiler,
  author       = {Liu, Amanda and Bernstein, Gilbert and Chlipala, Adam and Ragan-Kelley, Jonathan},
  title        = {{A Verified Compiler for a Functional Tensor Language}},
  journal      = {Proceedings of the ACM on Programming Languages},
  year         = {2024},
  volume       = {8},
  number       = {PLDI},
  articleno    = {160},
  pages        = {},
  numpages     = {23},
  publisher    = {Association for Computing Machinery},
  doi          = {10.1145/3656390},
  url          = {}
}

@article{Liu2026VerifiedSparseTensor,
  author       = {Liu, Amanda and Bernstein, Gilbert Louis and Kamil, Shoaib and Chlipala, Adam and Ragan-Kelley, Jonathan},
  title        = {{A Mechanized Algebra of Verified Data Structures for Optimizing Sparse Tensor Programs}},
  journal      = {Proceedings of the ACM on Programming Languages},
  year         = {2026},
  volume       = {10},
  number       = {PLDI},
  articleno    = {},
  pages        = {356--378},
  numpages     = {23},
  publisher    = {Association for Computing Machinery},
  doi          = {10.1145/3808261},
  url          = {}
}

@article{Lu2023Grisette,
  author       = {Lu, Sirui and Bod{\'i}k, Rastislav},
  title        = {{Grisette: Symbolic Compilation as a Functional Programming Library}},
  journal      = {Proceedings of the ACM on Programming Languages},
  year         = {2023},
  volume       = {7},
  number       = {POPL},
  articleno    = {16},
  pages        = {},
  numpages     = {33},
  publisher    = {Association for Computing Machinery},
  doi          = {10.1145/3571209},
  url          = {}
}

@article{Lu2025HieraSynth,
  author       = {Lu, Sirui and Bod{\'i}k, Rastislav},
  title        = {{HieraSynth: A Parallel Framework for Complete Super-Optimization with Hierarchical Space Decomposition}},
  journal      = {Proceedings of the ACM on Programming Languages},
  year         = {2025},
  volume       = {9},
  number       = {OOPSLA2},
  articleno    = {384},
  pages        = {},
  numpages     = {28},
  publisher    = {Association for Computing Machinery},
  doi          = {10.1145/3763162},
  url          = {}
}

@inproceedings{Li2024FTTN,
  author       = {Li, Xinyi and Li, Ang and Fang, Bo and Swirydowicz, Katarzyna and Laguna, Ignacio and Gopalakrishnan, Ganesh},
  title        = {{FTTN: Feature-Targeted Testing for Numerical Properties of NVIDIA \& AMD Matrix Accelerators}},
  booktitle    = {Proceedings of the 24th IEEE International Symposium on Cluster, Cloud and Internet Computing (CCGrid 2024)},
  year         = {2024},
  series       = {},
  volume       = {},
  pages        = {39--46},
  numpages     = {8},
  publisher    = {Institute of Electrical and Electronics Engineers},
  address      = {Piscataway, NJ, USA},
  doi          = {10.1109/ccgrid59990.2024.00014},
  url          = {}
}

@inproceedings{Laguna2022BinFPE,
  author       = {Laguna, Ignacio and Li, Xinyi and Gopalakrishnan, Ganesh},
  title        = {{BinFPE: Accurate Floating-Point Exception Detection for GPU Applications}},
  booktitle    = {Proceedings of the 11th ACM SIGPLAN International Workshop on the State of the Art in Program Analysis (SOAP 2022)},
  year         = {2022},
  series       = {},
  volume       = {},
  pages        = {1--8},
  numpages     = {8},
  publisher    = {Association for Computing Machinery},
  address      = {New York, NY, USA},
  doi          = {10.1145/3520313.3534655},
  url          = {}
}

@article{Laguna2023FPExceptions,
  author       = {Laguna, Ignacio and Tran, Anh and Gopalakrishnan, Ganesh},
  title        = {{Finding Inputs That Trigger Floating-Point Exceptions in Heterogeneous Computing via Bayesian Optimization}},
  journal      = {Parallel Computing},
  year         = {2023},
  volume       = {117},
  number       = {},
  articleno    = {},
  pages        = {103042},
  numpages     = {},
  publisher    = {Elsevier BV},
  doi          = {10.1016/j.parco.2023.103042},
  url          = {}
}

@inproceedings{Ikarashi2022Exo,
  author       = {Ikarashi, Yuka and Bernstein, Gilbert Louis and Reinking, Alex and Genc, Hasan and Ragan-Kelley, Jonathan},
  title        = {{Exocompilation for Productive Programming of Hardware Accelerators}},
  booktitle    = {Proceedings of the 43rd ACM SIGPLAN Conference on Programming Language Design and Implementation (PLDI 2022)},
  year         = {2022},
  series       = {},
  volume       = {},
  pages        = {703--718},
  numpages     = {16},
  publisher    = {Association for Computing Machinery},
  address      = {New York, NY, USA},
  doi          = {10.1145/3519939.3523446},
  url          = {}
}

@inproceedings{Ikarashi2025Exo2,
  author       = {Ikarashi, Yuka and Qian, Kevin and Droubi, Samir and Reinking, Alex and Bernstein, Gilbert Louis and Ragan-Kelley, Jonathan},
  title        = {{Exo 2: Growing a Scheduling Language}},
  booktitle    = {Proceedings of the 30th ACM International Conference on Architectural Support for Programming Languages and Operating Systems (ASPLOS 2025)},
  year         = {2025},
  series       = {},
  volume       = {1},
  pages        = {426--444},
  numpages     = {19},
  publisher    = {Association for Computing Machinery},
  address      = {New York, NY, USA},
  doi          = {10.1145/3669940.3707218},
  url          = {}
}

@inproceedings{Deng2025Coni,
  author       = {Deng, Wenqian and Liang, Jie and Wu, Zhiyong and Fu, Jingzhou and Wang, Mingzhe and Jiang, Yu},
  title        = {{Coni: Detecting Database Connector Bugs via State-Aware Test Case Generation}},
  booktitle    = {Proceedings of the 47th IEEE/ACM International Conference on Software Engineering (ICSE 2025)},
  year         = {2025},
  series       = {},
  volume       = {},
  pages        = {667--678},
  numpages     = {12},
  publisher    = {Institute of Electrical and Electronics Engineers},
  address      = {Piscataway, NJ, USA},
  doi          = {10.1109/icse55347.2025.00003},
  url          = {}
}

@inproceedings{Liang2023SequenceFuzzing,
  author       = {Liang, Jie and Chen, Yaoguang and Wu, Zhiyong and Fu, Jingzhou and Wang, Mingzhe and Jiang, Yu and Huang, Xiangdong and Chen, Ting and Wang, Jiashui and Li, Jiajia},
  title        = {{Sequence-Oriented DBMS Fuzzing}},
  booktitle    = {Proceedings of the 39th IEEE International Conference on Data Engineering (ICDE 2023)},
  year         = {2023},
  series       = {},
  volume       = {},
  pages        = {668--681},
  numpages     = {14},
  publisher    = {Institute of Electrical and Electronics Engineers},
  address      = {Piscataway, NJ, USA},
  doi          = {10.1109/icde55515.2023.00057},
  url          = {}
}

@inproceedings{Winterer2020SemanticFusion,
  author       = {Winterer, Dominik and Zhang, Chengyu and Su, Zhendong},
  title        = {{Validating SMT Solvers via Semantic Fusion}},
  booktitle    = {Proceedings of the 41st ACM SIGPLAN Conference on Programming Language Design and Implementation (PLDI 2020)},
  year         = {2020},
  series       = {},
  volume       = {},
  pages        = {718--730},
  numpages     = {13},
  publisher    = {Association for Computing Machinery},
  address      = {New York, NY, USA},
  doi          = {10.1145/3385412.3385985},
  url          = {}
}

@article{Wang2024DAInfer,
  author       = {Wang, Chengpeng and Zhang, Jipeng and Wu, Rongxin and Zhang, Charles},
  title        = {{DAInfer: Inferring API Aliasing Specifications from Library Documentation via Neurosymbolic Optimization}},
  journal      = {Proceedings of the ACM on Software Engineering},
  year         = {2024},
  volume       = {1},
  number       = {FSE},
  articleno    = {},
  pages        = {2469--2492},
  numpages     = {24},
  publisher    = {Association for Computing Machinery},
  doi          = {10.1145/3660816},
  url          = {}
}

@article{Chen2025FPConstraints,
  author       = {Chen, Qian and Cui, Chenqi and Gao, Fengjuan and Wang, Yu and Wang, Ke and Wang, Linzhang},
  title        = {{Solving Floating-Point Constraints with Continuous Optimization}},
  journal      = {Proceedings of the ACM on Programming Languages},
  year         = {2025},
  volume       = {9},
  number       = {PLDI},
  articleno    = {},
  pages        = {725--747},
  numpages     = {23},
  publisher    = {Association for Computing Machinery},
  doi          = {10.1145/3729279},
  url          = {}
}

@inproceedings{Mechtaev2018ReferenceRepair,
  author       = {Mechtaev, Sergey and Nguyen, Manh-Dung and Noller, Yannic and Grunske, Lars and Roychoudhury, Abhik},
  title        = {{Semantic Program Repair Using a Reference Implementation}},
  booktitle    = {Proceedings of the 40th IEEE/ACM International Conference on Software Engineering (ICSE 2018)},
  year         = {2018},
  series       = {},
  volume       = {},
  pages        = {129--139},
  numpages     = {11},
  publisher    = {Association for Computing Machinery},
  address      = {New York, NY, USA},
  doi          = {10.1145/3180155.3180247},
  url          = {}
}

@inproceedings{Saan2024WitnessValidation,
  author       = {Saan, Simmo and Schwarz, Michael and Erhard, Julian and Seidl, Helmut and Tilscher, Sarah and Vojdani, Vesal},
  title        = {{Correctness Witness Validation by Abstract Interpretation}},
  booktitle    = {Verification, Model Checking, and Abstract Interpretation (VMCAI 2024)},
  year         = {2024},
  series       = {Lecture Notes in Computer Science},
  volume       = {14499},
  pages        = {74--97},
  numpages     = {24},
  publisher    = {Springer Nature Switzerland},
  address      = {Cham, Switzerland},
  doi          = {10.1007/978-3-031-50524-9_4},
  url          = {}
}

@inproceedings{Flanagan2009FastTrack,
  author       = {Flanagan, Cormac and Freund, Stephen N.},
  title        = {{FastTrack: Efficient and Precise Dynamic Race Detection}},
  booktitle    = {Proceedings of the 30th ACM SIGPLAN Conference on Programming Language Design and Implementation (PLDI 2009)},
  year         = {2009},
  series       = {},
  volume       = {},
  pages        = {121--133},
  numpages     = {13},
  publisher    = {Association for Computing Machinery},
  address      = {New York, NY, USA},
  doi          = {10.1145/1542476.1542490},
  url          = {}
}

@inproceedings{Lee2023SeedAdaptive,
  author       = {Lee, Myungho and Cha, Sooyoung and Oh, Hakjoo},
  title        = {{Learning Seed-Adaptive Mutation Strategies for Greybox Fuzzing}},
  booktitle    = {Proceedings of the 45th IEEE/ACM International Conference on Software Engineering (ICSE 2023)},
  year         = {2023},
  series       = {},
  volume       = {},
  pages        = {384--396},
  numpages     = {13},
  publisher    = {Institute of Electrical and Electronics Engineers},
  address      = {Piscataway, NJ, USA},
  doi          = {10.1109/icse48619.2023.00043},
  url          = {}
}

@article{Xie2011MetamorphicML,
  author       = {Xie, Xiaoyuan and Ho, Joshua W.K. and Murphy, Christian and Kaiser, Gail and Xu, Baowen and Chen, Tsong Yueh},
  title        = {{Testing and Validating Machine Learning Classifiers by Metamorphic Testing}},
  journal      = {Journal of Systems and Software},
  year         = {2011},
  volume       = {84},
  number       = {4},
  articleno    = {},
  pages        = {544--558},
  numpages     = {15},
  publisher    = {Elsevier BV},
  doi          = {10.1016/j.jss.2010.11.920},
  url          = {}
}

@article{Li2026DLLens,
  author       = {Li, Meiziniu and Li, Dongze and Liu, Jianmeng and Cao, Jialun and Tian, Yongqiang and Cheung, Shing-Chi},
  title        = {{Enhancing Differential Testing with LLMs for Testing Deep Learning Libraries}},
  journal      = {ACM Transactions on Software Engineering and Methodology},
  year         = {2026},
  volume       = {35},
  number       = {4},
  articleno    = {88},
  pages        = {},
  numpages     = {39},
  publisher    = {Association for Computing Machinery},
  doi          = {10.1145/3735637},
  url          = {}
}

@inproceedings{He2022Qilin,
  author       = {He, Dongjie and Lu, Jingbo and Xue, Jingling},
  title        = {{Qilin: A New Framework for Supporting Fine-Grained Context-Sensitivity in Java Pointer Analysis}},
  booktitle    = {Proceedings of the 36th European Conference on Object-Oriented Programming (ECOOP 2022)},
  year         = {2022},
  series       = {Leibniz International Proceedings in Informatics},
  volume       = {222},
  pages        = {30:1--30:29},
  numpages     = {29},
  publisher    = {Schloss Dagstuhl -- Leibniz-Zentrum f{\"u}r Informatik},
  address      = {Dagstuhl, Germany},
  doi          = {10.4230/lipics.ecoop.2022.30},
  url          = {}
}

@inproceedings{Lagouvardos2020TensorShape,
  author       = {Lagouvardos, Sifis and Dolby, Julian and Grech, Neville and Antoniadis, Anastasios and Smaragdakis, Yannis},
  title        = {{Static Analysis of Shape in TensorFlow Programs}},
  booktitle    = {Proceedings of the 34th European Conference on Object-Oriented Programming (ECOOP 2020)},
  year         = {2020},
  series       = {Leibniz International Proceedings in Informatics},
  volume       = {166},
  pages        = {15:1--15:29},
  numpages     = {29},
  publisher    = {Schloss Dagstuhl -- Leibniz-Zentrum f{\"u}r Informatik},
  address      = {Dagstuhl, Germany},
  doi          = {10.4230/lipics.ecoop.2020.15},
  url          = {}
}

@misc{Reinking2022Halide,
  author       = {Reinking, Alex and Bernstein, Gilbert Louis and Ragan-Kelley, Jonathan},
  title        = {{Formal Semantics for the Halide Language}},
  year         = {2022},
  howpublished = {arXiv preprint},
  eprint       = {2210.15740},
  archivePrefix= {arXiv},
  doi          = {10.48550/arXiv.2210.15740},
  url          = {https://arxiv.org/abs/2210.15740},
  note         = {Version 1}
}

@inproceedings{Lee2020ADCorrectness,
  author       = {Lee, Wonyeol and Yu, Hangyeol and Rival, Xavier and Yang, Hongseok},
  title        = {{On Correctness of Automatic Differentiation for Non-Differentiable Functions}},
  booktitle    = {Advances in Neural Information Processing Systems (NeurIPS 2020)},
  year         = {2020},
  series       = {},
  volume       = {33},
  pages        = {6719--6730},
  numpages     = {12},
  publisher    = {Curran Associates, Inc.},
  address      = {Virtual Event},
  doi          = {},
  url          = {https://proceedings.neurips.cc/paper/2020/hash/4aaa76178f8567e05c8e8295c96171d8-Abstract.html}
}

@inproceedings{Lee2023MachineAD,
  author       = {Lee, Wonyeol and Park, Sejun and Aiken, Alex},
  title        = {{On the Correctness of Automatic Differentiation for Neural Networks with Machine-Representable Parameters}},
  booktitle    = {Proceedings of the 40th International Conference on Machine Learning (ICML 2023)},
  year         = {2023},
  series       = {Proceedings of Machine Learning Research},
  volume       = {202},
  pages        = {19094--19140},
  numpages     = {47},
  publisher    = {PMLR},
  address      = {Honolulu, Hawaii, USA},
  doi          = {},
  url          = {https://proceedings.mlr.press/v202/lee23p.html}
}

@article{Lee2023MixedPrecision,
  author       = {Lee, Wonyeol and Sharma, Rahul and Aiken, Alex},
  title        = {{Training with Mixed-Precision Floating-Point Assignments}},
  journal      = {Transactions on Machine Learning Research},
  year         = {2023},
  volume       = {},
  number       = {},
  articleno    = {},
  pages        = {},
  numpages     = {29},
  publisher    = {},
  doi          = {},
  url          = {https://openreview.net/forum?id=ZoXi7n54OB}
}

@inproceedings{Wang2024DataCoverage,
  author       = {Wang, Mingzhe and Liang, Jie and Zhou, Chijin and Wu, Zhiyong and Fu, Jingzhou and Su, Zhuo and Liao, Qing and Gu, Bin and Wu, Bodong and Jiang, Yu},
  title        = {{Data Coverage for Guided Fuzzing}},
  booktitle    = {Proceedings of the 33rd USENIX Security Symposium (USENIX Security 2024)},
  year         = {2024},
  series       = {},
  volume       = {},
  pages        = {2511--2526},
  numpages     = {16},
  publisher    = {USENIX Association},
  address      = {Philadelphia, PA, USA},
  doi          = {},
  url          = {https://www.usenix.org/conference/usenixsecurity24/presentation/wang-mingzhe}
}

@misc{Ge2026FlowMuT,
  author       = {Ge, Jiatao and Mu, Yanzhou and Zou, Yinglong and Hu, Yibo and Chen, Zhenyu and Fang, Chunrong},
  title        = {{Grounding Deep Learning Framework Testing in Data Flow}},
  year         = {2026},
  howpublished = {arXiv preprint},
  eprint       = {2610.15808},
  archivePrefix= {arXiv},
  doi          = {10.48550/arXiv.2610.15808},
  url          = {https://arxiv.org/abs/2610.15808v1},
  note         = {Version 1}
}

@misc{Fan2026ToolsCCBench,
  author       = {Fan, Haozhe and Wang, Wei and Liu, Xingchen and Liu, Man and Tian, Xingjian and Long, Haoquan and Liu, Zedong and Sun, Daran and Yang, Jinwu and Yang, Bo and Liu, Jie and Che, Yonggang and Zhao, Hairui and Tan, Guangming and Tao, Dingwen},
  title        = {{Tools-CC-Bench: A Benchmark Suite for Collective Communication with Compression in HPC and AI Workloads}},
  year         = {2026},
  howpublished = {arXiv preprint},
  eprint       = {2609.08739},
  archivePrefix= {arXiv},
  doi          = {10.48550/arXiv.2609.08739},
  url          = {https://arxiv.org/abs/2609.08739v2},
  note         = {Accepted at IISWC 2026; Version 2}
}

@misc{RondoArtifact,
  author       = {{Anonymous}},
  title        = {{Rondo: Stateful Testing for Deep Learning Compilers}},
  year         = {},
  howpublished = {Online repository},
  eprint       = {},
  archivePrefix= {},
  doi          = {},
  url          = {https://anonymous.4open.science/r/Rondo-108B/},
  note         = {Anonymized artifact}
}

\end{document}